\documentclass[english,11pt]{article}

\usepackage{multirow}
\usepackage{amsmath,amsthm}
\usepackage{amssymb}
\usepackage{graphicx}
\usepackage[dvipsnames]{xcolor}
\usepackage{xspace}
\usepackage{thmtools}
\usepackage{thm-restate}
\usepackage{enumerate}
\usepackage{hyperref}

\usepackage{fullpage}

\usepackage{tikz}
\usetikzlibrary{decorations.pathmorphing}
\usepackage{tikzscale}
\tikzstyle{none} = [draw=none]
\tikzstyle{small black dot}=[inner sep=0.5pt, fill=Black, draw=Black, shape=circle]
\tikzstyle{black dot}=[inner sep=1pt, fill=Black, draw=Black, shape=circle]
\tikzstyle{medium black dot}=[inner sep=1.5pt, fill=Black, draw=Black, shape=circle]
\tikzstyle{big black dot}=[inner sep=2pt, fill=Black, draw=Black, shape=circle]
\tikzstyle{tratt}==[-, dotted]
\tikzstyle{black}=[-, fill=none, draw=Black, line width=0.5pt]
\tikzset{snake it/.style={decorate, decoration=snake}}

\title{Approximating Optimal Labelings for Temporal Connectivity}
\date{}

\usepackage{authblk}

\author[1]{Daniele Carnevale}
\author[1]{Gianlorenzo D'Angelo}
\author[2]{Martin Olsen}

\affil[1]{Gran Sasso Science Institute, Italy\\

\texttt{daniele.carnevale@gssi.it, gianlorenzo.dangelo@gssi.it}}

\affil[2]{Aarhus University, Denmark\\

\texttt{martino@btech.au.dk}}

\newcommand{\Universe}{U\xspace}
\newcommand{\Sets}{\mathcal{C}\xspace}
\newcommand{\MAL}{\text{MAL}\xspace}
\newcommand{\MSL}{\text{MSL}\xspace}
\newcommand{\MASL}{\text{MASL}\xspace}
\newcommand{\DMAL}{\text{DMAL}\xspace}
\newcommand{\APX}{\text{APX}\xspace}

\newcommand{\DCSS}{\text{DCSS}\xspace}

\newtheorem{theorem}{Theorem}
\newtheorem{lemma}{Lemma}
\newtheorem{definition}{Definition}
\newtheorem{corollary}{Corollary}

\begin{document}

\maketitle

\begin{abstract}
In a temporal graph the edge set dynamically changes over time according to a set of time-labels associated with each edge that indicates at which time-steps the edge is available. Two vertices are connected if there is a path connecting them in which the edges are traversed in increasing order of their labels. We study the problem of scheduling the availability time of the edges of a temporal graph in such a way that all pairs of vertices are connected within a given maximum allowed time $a$ and the overall number of labels is minimized.

The problem, known as \emph{Minimum Aged Labeling} (\MAL), has several applications in logistics, distribution scheduling, and information spreading in social networks, where carefully choosing the time-labels can significantly reduce infrastructure costs, fuel consumption, or greenhouse gases.

The problem \MAL has previously been proved to be NP-complete on undirected graphs and \APX-hard on directed graphs. 
In this paper, we extend our knowledge on the complexity and approximability of \MAL in several directions. We first show that the problem cannot be approximated within a factor better than $O(\log n)$ when $a\geq 2$, unless $\text{P} = \text{NP}$, and a factor better than $2^{\log ^{1-\epsilon} n}$ when $a\geq 3$, unless $\text{NP}\subseteq \text{DTIME}(2^{\text{polylog}(n)})$, where $n$ is the number of vertices in the graph. Then we give a set of approximation algorithms that, under some conditions, almost match these lower bounds. In particular, we show that the approximation depends on a relation between $a$ and the diameter of the input graph.

We further establish a connection with a foundational optimization problem on static graphs called \emph{Diameter Constrained Spanning Subgraph} (\DCSS) and show that our hardness results also apply to \DCSS.
\end{abstract}

\section{Introduction}
We consider a scheduling problem on dynamic networks that is motivated by several applications in logistics, distribution scheduling, and information diffusion in social networks.

As a real-world example, consider a parcel-delivery scenario where there is a warehouse $W$ serving three cities as the center of a star topology, see Figure~\ref{fig:example} for an illustration.
Each city has a bunch of parcels to be delivered to the other cities and, for each pair of cities $(A,B)$, there is at least one parcel that must be delivered from $A$ to $B$.
There is a set of vehicles in $W$ that are used to deliver the parcels from the warehouse to the cities and vice-versa.
Each vehicle can possibly depart from $W$ at every hour. Upon arrival at one of the cities $A$, it delivers at $A$ the parcels destined to $A$ that were previously deposited at $W$, collects the parcels originating from $A$, and returns to the warehouse. 
A round trip of a vehicle from $W$ to one of the cities and back is called a \emph{trip}. For simplicity, we assume that the travel time is negligible. 
When a vehicle $V_1$ returns from city $A$ to the warehouse $W$, it deposits in $W$ all the parcels that depart from $A$ and whose final destination is any of the other cities. 
When a vehicle $V_2$ departs to another city $B$, it brings to $B$ all the parcels that have been deposited in $W$ so far and whose final destination is $B$. 
If the trip of $V_1$ is scheduled earlier than that of $V_2$, the parcels directed from $A$ to $B$ are delivered. Otherwise, these parcels must wait for the next trip.
Carefully scheduling the trips of all vehicles might reduce, at the same time, the delivery time and the costs in terms of number of required trips, vehicles used, fuel consumption, pollutants, and emission of greenhouse gases.
For example, the 1\textsuperscript{st} schedule in Figure~\ref{fig:example}, requires 6 trips to deliver all parcels and the last parcels are delivered at 9am. 
The 2\textsuperscript{nd} schedule optimizes the latest arrival time and ensures that the last parcels are delivered two hours earlier than the previous solution, at 7am, but still requires 6 trips to deliver all parcels. The 3\textsuperscript{rd} schedule, instead, optimizes the required number of trips, reducing them to 5, but the last parcels are delivered at 8am, one hour later than the previous case.

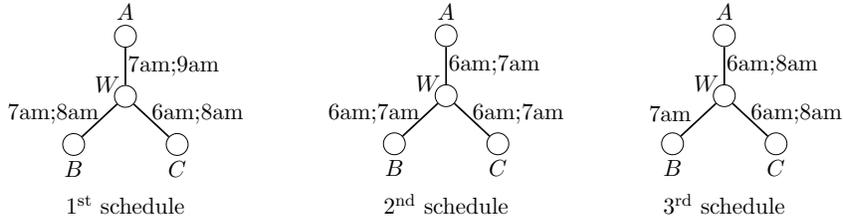
\begin{figure}[t]
    \centering
    \scalebox{0.8}{
  \begin{tikzpicture}
  \node[circle,draw] (W) at (0,0) {};
  \node (w) at (-0.3,0.2) { $W$};
  \node[circle,draw] (A) at (0,1) {};
  \node (a) at (0,1.4) { $A$};
  \node[circle,draw] (B) at (-0.86,-0.8) {};
  \node (b) at (-0.86,-1.2) { $B$};
  \node[circle,draw] (C) at (0.86,-0.8) {};
  \node (c) at (0.86,-1.2) { $C$};
  \draw[-,line width=0.8pt] (W) to (A);
  \node (aw) at (0.8,0.5) {7am;9am};
  \draw[-,line width=0.8pt] (W) to (B);
  \node (bw) at (-1.2,-0.3) {7am;8am};
  \draw[-,line width=0.8pt] (W) to (C);
  \node (cw) at (1.2,-0.3) { 6am;8am};
  \node (cap) at (0,-1.8) { 1\textsuperscript{st} schedule};

\end{tikzpicture}\quad\quad\quad
  \begin{tikzpicture}
  \node[circle,draw] (W) at (0,0) {};
  \node (w) at (-0.3,0.2) { $W$};
  \node[circle,draw] (A) at (0,1) {};
  \node (a) at (0,1.4) { $A$};
  \node[circle,draw] (B) at (-0.86,-0.8) {};
  \node (b) at (-0.86,-1.2) { $B$};
  \node[circle,draw] (C) at (0.86,-0.8) {};
  \node (c) at (0.86,-1.2) { $C$};
  \draw[-,line width=0.8pt] (W) to (A);
  \node (aw) at (0.8,0.5) {6am;7am};
  \draw[-,line width=0.8pt] (W) to (B);
  \node (bw) at (-1.2,-0.3) {6am;7am};
  \draw[-,line width=0.8pt] (W) to (C);
  \node (cw) at (1.2,-0.3)  {6am;7am};
  \node (cap) at (0,-1.8) { 2\textsuperscript{nd} schedule};
\end{tikzpicture}\quad\quad\quad
  \begin{tikzpicture}
  \node[circle,draw] (W) at (0,0) {};
  \node (w) at (-0.3,0.2) { $W$};
  \node[circle,draw] (A) at (0,1) {};
  \node (a) at (0,1.4) { $A$};
  \node[circle,draw] (B) at (-0.86,-0.8) {};
  \node (b) at (-0.86,-1.2) { $B$};
  \node[circle,draw] (C) at (0.86,-0.8) {};
  \node (c) at (0.86,-1.2) { $C$};
  \draw[-,line width=0.8pt] (W) to (A);
  \node (aw) at (0.8,0.5) {6am;8am};
  \draw[-,line width=0.8pt] (W) to (B);
  \node (bw) at (-0.9,-0.3) {7am};
  \draw[-,line width=0.8pt] (W) to (C);
  \node (cw) at (1.2,-0.3)  {6am;8am};
  \node (cap) at (0,-1.8) { 3\textsuperscript{rd} schedule};
\end{tikzpicture}}
    \caption{The scheduled time of trips are reported close to the edges.
    \textbf{1\textsuperscript{st} schedule}: All the 6 scheduled trips are needed to deliver all parcels and the last parcel is delivered at 9am. 
    \textbf{2\textsuperscript{nd} schedule}: All the parcels are deposited in $W$ with the 3 trips at 6am and then are delivered from $W$ to the right cities at 7am; still  6 trips are needed but the last parcel is delivered at 7am. 
    \textbf{3\textsuperscript{rd} schedule}: The parcels leaving cities $A$ and $C$ are deposited in $W$ with the trips at 6am; the single trip to $B$ at 7am brings the parcels directed to $B$ previously deposited at 6am and deposits the parcels from $B$ in $W$; finally, the trips at 8am deliver the parcels directed to $A$ and $C$. Only 5 trips are needed, but the last parcel is delivered at 8am.}\label{fig:example}
\end{figure}

The problem of scheduling trips becomes much more complex if we consider a general network in which each vertex might serve both as a warehouse and as a city, and connections among vertices are given by an arbitrary underlying graph. Moreover, similar problems arise in other contexts such as distribution scheduling~\cite{DP20} and information spreading, where the aim is to schedule a small number of meetings among employees of a company in such a way that each employee can share its own information with any other by a given time, see~\cite{DES}.

Motivated by these applications, we consider the following question: \emph{What is the minimum number of trips needed to deliver all parcels within a given time?}

We model a schedule of trips along edges of a network with a \emph{temporal graph} (a.k.a. dynamic graph) in which the scheduling time of a vehicle is represented as an edge-label and a path in a graph is valid (or \emph{temporal}) only if the edges are traversed in increasing order of their labels.  We then consider the optimization problem of assigning the minimum number of labels to the edges of a graph in such a way that each pair of vertices is connected via a temporal path and the largest label is not greater than a given integer $a$, called the \emph{maximum allowed age}.
This problem, called \emph{Minimum Aged Labeling} (\MAL), has been introduced by Mertzios et al.~\cite{MMS}, who proved that it is \APX-hard in \emph{directed graphs}.
Later, Klobas et al.\cite{KMMS} showed that \MAL is NP-complete also for \emph{undirected graphs}.
To the best of our knowledge, there are no hardness or algorithmic results on the approximation of \MAL in undirected graphs. Moreover, the reduction used to prove the \APX-hardness in directed graphs~\cite{MMS} cannot be easily adapted to undirected graphs as in the constructed graph the direction of edges is used to bound the reachability of vertices. In this paper, we study the complexity of approximating \MAL in undirected graphs showing when, depending on a relation between $a$ and the diameter $D_G$ of the input graph, it is hard to approximate, or it can be approximated in polynomial-time.

\MAL also has a theoretical motivation as it can be interpreted as a \emph{dynamic version} of a foundational graph theoretical problem called the \emph{Diameter Constrained Spanning Subgraph} (\DCSS) problem, which asks to find a spanning subgraph $H$ of a graph $G$ such that the diameter of $H$ is at most a given integer and its number of edges is minimized.

\paragraph*{Our results.}
We provide both hardness of approximation lower bounds and approximation algorithms for \MAL.

\noindent
\emph{Hardness of approximation.}
We first prove that, even when the maximum allowed age $a$ of a labeling is a fixed value greater or equal to 2, \MAL cannot be approximated within a factor better than $O(\log n)$, unless $\text{P} = \text{NP}$.
Then, we show that, unless $\text{NP}\subseteq \text{DTIME}(2^{\text{polylog}(n)})$, we cannot find any  $2^{\log ^{1-\epsilon} n}$-approximation algorithm for \MAL, for any $\epsilon\in (0,1)$, even when $a$ is a fixed value greater or equal to 3.

These results advance our knowledge on the computational complexity of \MAL in two directions.
(1)
From an exact computation point of view, the NP-hardness given in~\cite{KMMS} only holds for $a=D_G=10$, while we show that \MAL is NP-hard for any fixed $a\geq 2$ (still, we require $a=D_G$). This closes the characterization of the computational complexity of \MAL with respect to the parameter $a$, as the case $a=1$ is trivial. Moreover, we provide a considerably simpler reduction than the one in~\cite{KMMS}.
(2)
From an approximation point of view, the reduction in~\cite{MMS} shows that \MAL is \APX-hard for $a=D_G=9$ in directed graphs. We show two stronger lower bounds on the approximation, namely that \MAL is hard to approximate better than a logarithmic factor, under $\text{P} \neq \text{NP}$, and a factor $2^{\log ^{1-\epsilon} n}$, under a stronger complexity condition. The second lower bound suggests that it is unlikely to approximate \MAL to a factor better than a polynomial.
Moreover, our lower bounds hold even for any fixed $a\geq 2$ and for $a\geq 3$, respectively, and for undirected graphs (again, we require $a=D_G$). Finally, we remark that our lower bounds also apply to \DCSS.

Our hardness results are given in Section~\ref{sec:hardness} and summarized in Table~\ref{tab:hardness}.

\begin{table}[t]
\caption{Summary of the approximation hardness results for \MAL (see Section~\ref{sec:hardness})}
\label{tab:hardness}
\centering
\begin{tabular}{c|ccc|}
\cline{2-4}
 & \multicolumn{1}{c|}{value of $a$} & \multicolumn{1}{c|}{Complexity assumption} & Approximation hardness \\ \hline
\multicolumn{1}{|c|}{\multirow{3}{*}{$a=D_G$}} & \multicolumn{1}{c|}{$a=2$} & \multicolumn{1}{c|}{$\text{P}\!\not =\!\text{NP}$} & $\Theta( \log n)$ \\ \cline{2-4} 
\multicolumn{1}{|c|}{} & \multicolumn{1}{c|}{\multirow{2}{*}{$a\geq 3$}} & \multicolumn{1}{c|}{$\text{P}\!\not=\!\text{NP}$} & $\Omega( \log n)$ \\ \cline{3-4} 
\multicolumn{1}{|c|}{} & \multicolumn{1}{c|}{} & \multicolumn{1}{c|}{$\text{NP}\not \subseteq \text{DTIME}(2^{\text{polylog}(n)})$} & $\Omega(2^{\log ^{1-\epsilon} n})$ \\ \hline
\multicolumn{1}{|c|}{$D_G < a < 2 D_G+2$} & \multicolumn{3}{c|}{Open} \\ \hline
\multicolumn{1}{|c|}{$a \geq 2 D_G +2$} & \multicolumn{3}{c|}{Polynomial-time solvable} \\ \hline
\end{tabular}
\end{table}

\noindent
\emph{Approximation algorithms.}
Like in~\cite{KMMS} and~\cite{MMS}, all our reductions require that $a=D_G$. Hence, we investigate the approximability of \MAL when $a\geq D_G$, addressing an open question posed in~\cite{KMMS}. We give three sets of results, which interestingly show how the approximation of \MAL  depends on a relation between $a$ and $D_G$.

(1) We first consider the case in which $a$ is sufficiently larger than $R_G$, the radius of $G$. If $a\geq 2R_G$ ($a\geq 2R_G +1$, respectively), we can compute in polynomial-time a solution that requires at most only 2 labels (1 label, respectively) more than the optimum. Observe that these additive approximation bounds correspond to asymptotic multiplicative bounds with approximation factors that arbitrarily approach 1 as the input size grows. Moreover, if $a\geq 2D_G+2$, we can compute an optimal solution in polynomial-time. As \MAL does not admit any feasible solution when $a<D_G$, this first set of results leaves open the cases when $D_G \leq a < 2R_G$.

(2) We then consider the case in which $a$ is slightly larger than $D_G$ by exploiting a relation between \MAL and \DCSS. Specifically, we show that there is a gap of a factor $a$ between the approximation of \MAL and that of \DCSS. We first prove that when $a=D_G=2$, \MAL can be approximated with a logarithmic factor, which asymptotically matches our first hardness lower bound.\footnote{When $a=2$ and $D_G=1$ the graph is a clique. Hence $R_G=1$ and we can solve \MAL with 2 labels more than the optimum.} When $a\geq D_G+2$, we achieve an approximation factor of $O(D_G\cdot n^{1/2})$, which is sublinear when $D_G$ is sufficiently small. Moreover, if $a\geq D_G+4$ or $a\geq D_G+6$, we reduce the approximation factor to $O(D_G\cdot n^{2/5})$ and $O(D_G\cdot n^{1/3})$, respectively. Finally, for any value of $a \geq D_G$ we achieve an approximation factor of $O(D_G\cdot n^{3/5+\epsilon})$, for any $\epsilon>0$. All these approximation factors linearly depend  on $D_G$, due to the gap between the approximation of \MAL and that of \DCSS.

(3) Our main algorithmic contribution consists in approximating \MAL when $D_G \leq a < 2R_G$ without passing through  \DCSS, thus avoiding a linear dependency on $D_G$ in the approximation ratio. We show that when $a \geq \lceil 3/2 \cdot D_G \rceil$ ($a \geq \lceil 5/3 \cdot D_G \rceil$, respectively), then we can approximate \MAL within a factor of $O(\sqrt{n\log n})$ ($O(\sqrt[3]{D_G n\log^2n})$, respectively). Both bounds are sublinear, and the second algorithm outperforms the first one  when $D_G=o(\sqrt{n/\log n})$ but requires greater values of $a$.

Our approximation results for \MAL are given in Section~\ref{sec:alg} summarized in Table~\ref{tab:approx}. The approximation results for \DCSS are given in Section~\ref{sec:dcss}.

\begin{table}[t]
\caption{Summary of the approximation results for \MAL as a function of the age bound $a$ and the graph diameter $D_G$ (see Section~\ref{sec:alg})}
\label{tab:approx}
\centering
\resizebox{\columnwidth}{!}{
\begin{tabular}{|c|c|c|}
\hline
\multicolumn{1}{|c|}{$a$}   & \multicolumn{1}{c|}{$D_G$} & \multicolumn{1}{c||}{Approximation ratio}  \\ \hline
$a\leq D_G +1$  & $D_G\in o(n^{2/5})$ & $O(D_G \cdot n^{3/5+\epsilon})$ \\ \hline
$D_G+2 \leq a \leq D_G+3$ & $D_G\in o(n^{1/2})$ & $O(D_G \cdot n^{1/2})$  \\ \hline
$D_G+4 \leq a \leq D_G+5$ & $D_G\in o(n^{3/5})$ & $O(D_G \cdot n^{2/5})$  \\ \hline
$D_G+6 \leq a < \lceil 3/2D_G\rceil $ & $D_G\in o(n^{2/3})$ & $O(D_G \cdot n^{1/3})$ \\ \hline
\multirow{2}{*}{$\lceil 3/2D_G\rceil \leq a < \lceil 5/3D_G\rceil $} & $D_G\in \Omega(\sqrt{n^{1/3}\cdot \log n})$  & $O(\sqrt{n \cdot \log n})$  \\ \cline{2-3} 
                          & otherwise & $O(D_G \cdot n^{1/3})$    \\ \hline
\multirow{3}{*}{$a \geq \lceil 5/3D_G\rceil $}         & $D_G\in \Omega(\sqrt{n/ \log n})$ & $O(\sqrt{n \cdot \log n})$ \\ \cline{2-3} 
                          & $D_G \in \Omega(\log n) \land D_G\in O(\sqrt{n/ \log n}) $ & $O(\sqrt[3]{D_G\cdot n \cdot \log^2 n})$ \\ \cline{2-3}
                          & otherwise & $O(D_G \cdot n^{1/3})$  \\ \hline
\end{tabular}
}
\end{table}

\paragraph*{Related Work.}
Due to their versatility, temporal graphs have been considered from several perspectives and using different terminology, e.g. \emph{dynamic}, \emph{evolving}, and \emph{time-varying} graphs or networks, see~\cite{M2016}. Mainly motivated by virus-spread minimization (see e.g. \cite{BI16,EK18}), an area that received considerable interest is the one related to the modification of a temporal network in such a way that some objective is optimized (see a recent survey \cite{Meeks22}). Several operations were considered e.g. delaying labels and merging consecutive times~\cite{DP20}, edge-deletion and label-deletion~\cite{EMMZ21}, and changing the relative order of times ~\cite{EMS21}. Moreover, Molter et al.~\cite{MRZ24} studied how the choice between edge-deletion and delaying influences the parameterized complexity of the reachability-minimization objective.
In~\cite{DES}, the authors studied a problem similar yet orthogonal to \MAL, where the objective is to minimize the maximum time required by a subset of vertices to reach every other vertex, through label shifting. Klobas et al.~\cite{KMMS}, also considered a generalization of \MAL where only a subset of \emph{terminal} vertices must be connected, subject to a constraint on the maximum allowed age, and show that the problem is $W[1]$-hard when parameterized by the number of labels.


\section{Preliminaries}

For an integer $k\in {\mathbb{N}}$, let $[k]:=\{1,2,\ldots,k\}$ and $[k]_0:= \{0\} \cup[k]$. We consider simple undirected graphs $G=(V,E)$ with $n=|V|$ vertices and $m=|E|$ edges. We also use $V(G)$ and $E(G)$ to refer to the vertex set and the edge set of graph $G$, respectively. We denote an edge $e\in E$ between two vertices $v$ and $w$ with $e=\{v,w\}$, and say that $v$ and $w$ are the \emph{endpoints} of $e$ and that $v$ (as well as $w$) is incident to $e$. Let $v$ and $w$ be two vertices of a graph $G$, we say that $v$ and $w$ are \emph{adjacent} if $\{v,w\}\in E$. We denote the degree of a vertex $v$ with $\deg_G(v)$, representing the number of distinct edges of $G$ incident to $v$.

A \emph{subgraph} $H = (V', E')$ of a graph $G$ (denoted $H \subseteq G$) is a graph such that $V' \subseteq V(G)$ and $E' \subseteq E(G)$, with the additional condition that if ${v, w} \in E'$, then $v, w \in V'$. If $V' = V(G)$ and $H$ is connected, then $H$ is called a \emph{spanning subgraph}.

A \emph{path} $P$ between vertices $v$ and $w$ in graph $G$ is a sequence of distinct vertices $v_0,v_1,\ldots,v_k\in V$ and sequence of distinct edges $\{v_i,v_{i+1}\} \in E$, for each $i\in [k-1]_0$, where $v_0=v$ and $v_k=w$. The \emph{length} of $P$ is equal to the number of its edges, i.e. $k$. 
A \emph{shortest path} is a path with minimum length. A \emph{cycle} in $G$ is a path of length at least 2, plus an edge between $v_0$ and $v_k$. $C_k$ is a cycle with $k$ edges.

The \emph{distance} between two vertices $v$ and $w$ in the graph $G$ is denoted $d_G(v,w)$ and is equal to the length of a shortest path between $v$ and $w$. Let $S\subseteq V$, we denote by $d_G(v,S)$ the distance between $v$ to its closest vertex in $S$, i.e. $d_G(v,S)=\min_{s\in S} d_G(v,s) $. A graph $G$ is \emph{connected} if there exists a path between any pair of vertices.
A \emph{Shortest Path Forest} (SPF) $F$ rooted at some set of vertices $S\subseteq V$ is a forest spanning all vertices $V$ such that $d_F(w,S) = d_G(w,S)$ , for each $w\in V(G)$. If $S$ is a singleton, $S=\{v\}$, a SPF rooted at $S$ is called a \emph{Shortest Path Tree} (SPT) rooted at $v$.

Given a vertex $v\in V$, the \emph{eccentricity} $\text{ecc}_G(v)$ of $v$ in $G$ is the maximum distance between $v$ and any other vertex, i.e. $\text{ecc}_G(v)=\max_{w\in V} d_G(v,w)$. The \emph{diameter} and \emph{radius} of $G$ are defined as $D_G=\max_{v\in V} \text{ecc}_G(v)$ and $R_G=\min_{v\in V} \text{ecc}_G(v)$, respectively. We say that a vertex $v$ is a \textit{center} of a graph $G$ if $\text{ecc}_G(v)=R_G$.

We will drop the subscript $G$ in the notation when the graph in question is clear from the context.

A \textit{temporal graph} is a pair $(G,\lambda)$ where $G$ is a graph and $\lambda : E \longrightarrow 2^{\mathbb{N}\setminus \{0\}}$ is a function assigning to every edge of $G$ a set of discrete \textit{labels} that we interpret as presence in time. 
The \textit{lifetime} (a.k.a. \emph{age}) $L_\lambda$ of a temporal graph $(G,\lambda)$ is the largest label in $\lambda$ i.e. $L_\lambda=\max \{t:  t\in \lambda(e) , e \in E\}$.  The total number of labels in a temporal graph is $|\lambda|=\sum_{e\in E}|\lambda(e)|$.

One central notion in temporal graphs is that of a \textit{temporal path}. In this paper, we adopt the definition commonly referred to as a \emph{strict} temporal path. A \textit{temporal path} is a sequence $(e_1,t_1),(e_2,t_2),\dots,(e_k,t_k)$ where $e_1,e_2,\dots,e_k$ is a path in $G$, $t_i\in \lambda(e_i)$ for every $i\in[k]$, and $t_1<t_2<\dots<t_k$. We say that a vertex $v$ is \textit{temporally reachable} from a vertex $w$ if there exists a temporal path from $w$ to $v$.
A temporal graph is called \textit{temporally connected} if every vertex is temporally reachable from every other vertex.

We are ready to introduce the \emph{Minimum Aged Labeling} (MAL) problem, which is the subject of this paper.
\begin{definition}
    Given a graph $G=(V,E)$ and an integer $a$, the Minimum Aged Labeling (\text{MAL}) problem asks to find a function $\lambda : E\rightarrow 2^{\mathbb{N}\setminus\{0\}}$ such that $(G,\lambda)$ is temporally connected, $L_\lambda \leq a$, and $|\lambda|$ is minimized.
\end{definition}
We assume that the graph $G$ is connected as otherwise \MAL does not admit any feasible solution.
The same holds if $a< D_G$, as already observed in~\cite{KMMS}.
On the other hand, if $a\geq 2D_G+2$, then MAL is solvable in polynomial-time as we will prove in Theorem~\ref{thm:malradiusdiameter}-($iii$) by using a result from~\cite{KMMS}. Hence, we will consider in the rest of the paper the case $D_G \leq a \leq 2D_G+2$.

A problem that is strictly related to \MAL is the \emph{Diameter Constrained Spanning Subgraph} (\DCSS) problem.
\begin{definition}
    Given a graph $G$ and an integer $d\geq D_G$, the \DCSS problem asks to find a spanning subgraph $H$ such that $D_H\leq d$ and the number of edges of $H$ is minimized. 
\end{definition}

There is a close connection between problems \DCSS and \MAL. Roughly speaking, \MAL can be considered a temporal version of the \DCSS problem. First, we will state a technical lemma and then a theorem that relates the approximations of \MAL and \DCSS.
\begin{lemma}\label{lem:mal-dss}
    Let $G$ be a graph and $b$ be an integer such that $D_G\leq b\leq 2D_G+2$.
    \begin{enumerate}[i)]
        \item Given a feasible solution $H$ to an instance $(G,b)$ of \DCSS, we can compute in polynomial-time a feasible solution $\lambda$ to the instance $(G,b)$ of MAL with $|\lambda| = b\cdot |E(H)|$.
        \item Given a feasible solution $\lambda$ to an instance $(G,b)$ of MAL, we can compute in polynomial-time a feasible solution $H$ to the instance $(G,b)$ of \DCSS with $|E(H)| \leq |\lambda|$.
    \end{enumerate}
\end{lemma}
\begin{proof}
    $i$) Let $H$ be a feasible solution to the instance $(G,b)$ of \DCSS, we construct $\lambda$ by assigning to every edge in $E(H)$ the set of labels $[b]$. This procedure requires polynomial-time in $G$ since $b=O(|V(G)|)$. To see that $(G,\lambda)$ is temporally connected, let $v,w$ be any pair of vertices of $G$. As $H$ is a feasible solution for the instance $(G,b)$ of \DCSS, then     $D_H\leq b$, which implies that there must exist a path $e_1,e_2,\dots,e_k$ between $v$ and $w$ of length $k\leq b$. Observe that $(e_1,1),(e_2,2),\dots,(e_k,k)$ is a valid temporal path in $(G,\lambda)$ from $v$ to $w$. Thus, $(G,\lambda)$ is temporally connected.
    
    $ii$) Let $\lambda$ be a feasible solution to the instance $(G,b)$ of MAL, we construct $H$ by removing the labels from the edges of $(G,\lambda)$. Formally, $E(H):= \{e : e\in E(G) \land \lambda(e)\not = \emptyset\}$ and observe that $|E(H)|\leq |\lambda|$ by definition. Since $(G,\lambda)$ is temporally connected with lifetime at most $b$, the graph $H$ is connected and has diameter at most $b$. In fact, any temporal path that uses at most $b$ distinct labels, has at most $b$ edges.
\end{proof}
We can now relate the approximations of \MAL and \DCSS through the following theorem.
\begin{theorem}\label{thm:apx-mal-dss}
    Let $G$ be a graph and $b$ be an integer such that $b=O(|V(G)|)$.
       \begin{enumerate}[i)]
       \item If there exists an $\alpha$-approximation algorithm for \MAL, then there exists an $(\alpha b)$-approximation algorithm for \DCSS, where $(G,b)$ is the input to \DCSS.
       \item If there exists an $\alpha$-approximation algorithm for \DCSS, then there exists an $(\alpha b)$-approximation algorithm for \MAL, where $(G,b)$ is the input to \MAL.
       \end{enumerate}
\end{theorem}

\begin{proof}
    $i$) Let $(G,b)$ be an instance of \DCSS. Let us consider $(G,b)$ as an instance of \MAL and apply to it an $\alpha$-approximation algorithm for \MAL. Let $\lambda$ be the returned solution, we have $|\lambda| \leq \alpha |\lambda^*|$, where $\lambda^*$ is an optimal solution to the \MAL instance. By Lemma~\ref{lem:mal-dss}-($ii$) we can compute in polynomial-time a solution $H$ for the instance of \DCSS such that $|E(H)|\leq |\lambda| \leq \alpha |\lambda^*|$.
    
    Let $H^*$ be an optimal solution to the \DCSS instance. By Lemma~\ref{lem:mal-dss}-($i$), starting from $H^*$ we can construct a solution $\lambda'$ for the corresponding \MAL instance such that $|\lambda'|= b |E(H^*)|$ and then we have $|\lambda^*|\leq  b |E(H^*)|$. It follows that $|E(H)|\leq\alpha b |E(H^*)|$.

    $ii$) Let $(G,b)$ be an instance of \MAL. Let us consider $(G,b)$ as an instance of \DCSS and apply to it an $\alpha$-approximation algorithm for \DCSS. Let $H$ be the returned solution, we have $|E(H)| \leq \alpha |E(H^*)|$, where $H^*$ is an optimal solution to the \DCSS instance. By Lemma~\ref{lem:mal-dss}-($i$) we can compute in polynomial-time a solution $\lambda$ for the instance of \MAL such that $|\lambda| = b|E(H)| \leq b\alpha |E(H^*)|$.  By Lemma~\ref{lem:mal-dss}-($ii$), we have that $|E(H^*)|\leq |\lambda^*|$, which implies that $|\lambda| \leq b\alpha|\lambda^*|$.
\end{proof}

\section{Hardness of Approximation}\label{sec:hardness}
In this section, we prove that \MAL and \DCSS are hard to approximate, even under some restrictions on the input instance. 

We begin by presenting our logarithmic lower bound for $a=2$ in the next theorem. 
\begin{restatable}{theorem}{malinapprox}\label{thm:malinapprox}
For every $\epsilon\in(0,1/4)$, there is no polynomial-time $(\epsilon \log |V|)$-approximation algorithm for \MAL, unless $\text{P} = \text{NP}$. The hardness holds even when the maximum allowed age $a$ is equal to 2.
\end{restatable}
Observe that Theorem~\ref{thm:apx-mal-dss} together with Theorem~\ref{thm:malinapprox} imply that a logarithmic lower bound also holds for \DCSS when $d=2$.

To prove Theorem~\ref{thm:malinapprox}, we introduce the \emph{Set Cover} (SC) problem.
In SC we are given a universe $U=\{u_1,\dots,u_\eta\}$, with $|U|=\eta$, and a collection of $\mu$ subsets of $U$, $\mathcal{C} = \{s_1,\dots,s_\mu \}\subseteq 2^U$. 
The goal is to find a minimum-size subcollection $\mathcal{C}^*\subseteq \mathcal{C}$ such that $\bigcup_{s_i\in\mathcal{C}^* }s_i=U$. 

The next theorem proves the intractability of approximating SC to a factor smaller than $\log N$, where $N$ is the size of an instance.

\begin{theorem}[\cite{DinurSteurer}, Corollary 1.5]\label{lemma:sc}
    For every $\alpha \in(0,1)$, it is \emph{NP}-hard to approximate \emph{SC} to within $(1-\alpha)\log N$, where $N$ is the size of the instance. The reduction runs in time $N^{O(1/\alpha)}$.
\end{theorem}

\begin{figure}[t]
    \centering
    \pgfdeclarelayer{nodelayer}
    \pgfdeclarelayer{edgelayer}
    \pgfsetlayers{nodelayer,edgelayer}
    \includegraphics[width=0.5 \textwidth]{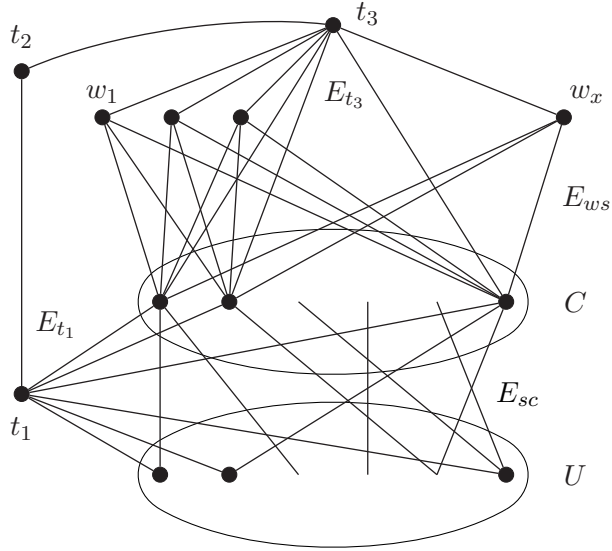}
    \caption{The graph $G$ constructed starting from a SC instance $(U,\mathcal{C})$ in the proof of Theorem~\ref{thm:malinapprox}.}
    \label{fig1}
\end{figure}

Let $(U,\mathcal{C})$ be an instance of the SC problem. We construct a graph $G=(V,E)$ which we will use to prove Theorem~\ref{thm:malinapprox}. See Figure~\ref{fig1} for a reference of the construction.
Graph $G$ contains a vertex $s_j$ for each set of $\mathcal{C}$, a vertex $u_i$ for each element of the universe $U$, and two sets of  new vertices $T=\{t_1,t_2,t_3\}$ and $W=\{w_1,w_2,\dots,w_x\}$ where $x$ is an integer to be defined later. Formally, $V(G)= \mathcal{C}\cup U\cup T \cup W$.
The set of edges of $G$ contains the edges induced by the instance of SC, that is $E_{\text{sc}}=\{\{u_i,s_j\} \mid u_i\in U \land s_j\in \mathcal{C} \land u_i\in s_j\}$. It also contains edges $E_{t_1}=\{\{t_1,z\}\mid z\in U\cup \mathcal{C} \cup \{t_2\}\}$ and the edges $E_{t_3}=\{\{t_3,z\}\mid z\in W \cup \mathcal{C} \cup \{t_2\}\}$. Moreover, we add to $G$ the edges between each $w_l\in W$ and all vertices in $\mathcal{C}$, more formally $E_{\text{ws}}=\{\{w_l,s_j\}:w_l \in W \land s_j\in \mathcal{C}\}$. Finally, the edge set of $G$ is equal to $E(G)=E_{\text{sc}}\cup E_{t_1} \cup E_{t_3}\cup E_{\text{ws}}$.
The reader can check that the diameter of $G$ is equal to 2.

The idea behind the reduction is that, since all paths of length 2 between any $ w \in W $ and any $ u \in U $ pass through $ C $, we must select two (possibly identical) set covers for each $ w \in W $—one to allow $ w $ to reach all vertices in $ U $ and another to allow all vertices in $ U $ to reach $ w $. The vertices in $ T $ help connect the other pairs of vertices. However, the labels of the edges incident to the vertices in $ T $ might contribute a dominant term to the solution size. We address this by adding a sufficient number of vertices in $ W $.

Before proving Theorem~\ref{thm:malinapprox}, we need the following lemma.

\begin{lemma}\label{lem:mal2}
    Let $(U,\mathcal{C})$ be an instance of SC and let $G$ be the graph constructed as described above. Given an optimal set cover of size $k^*$ for $(U,\mathcal{C})$, we can construct a feasible solution $\lambda$ to the instance $(G,2)$ of MAL with $|\lambda|\leq 6x+2xk^*$.
\end{lemma}
\begin{proof}

    Let $\mathcal{C}^* \subseteq \mathcal{C}$ be a set cover of optimal size equal to $k^*$. We construct $\lambda$ as follows (see Figure~\ref{fig:rid_mal_2}).
    \begin{figure}[t]
        \centering
        \pgfdeclarelayer{nodelayer}
        \pgfdeclarelayer{edgelayer}
        \pgfsetlayers{nodelayer,edgelayer}
        \includegraphics[width=0.5 \textwidth]{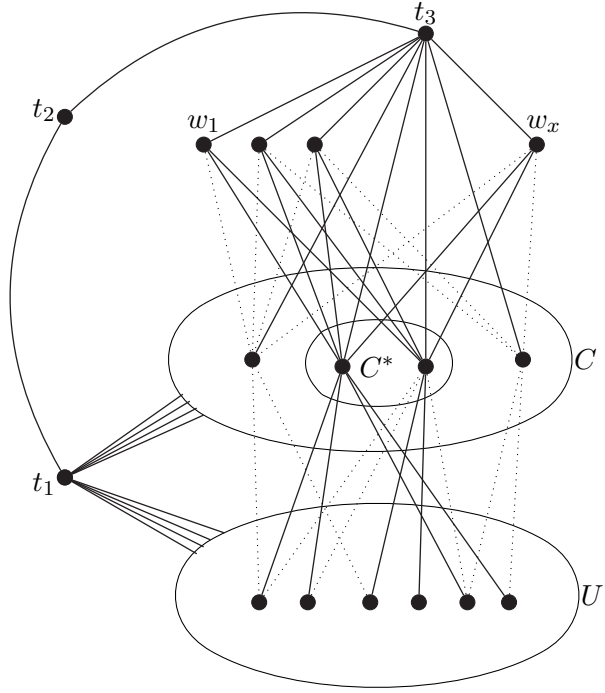}
        \caption{The labeling $\lambda$ constructed from an optimal set cover in the proof of Lemma~\ref{lem:mal2}. Each non-dashed edge in the \MAL feasible solution is assigned the labels $\{1,2\}$.}
        \label{fig:rid_mal_2}
    \end{figure}
    For all $e\in E_{t_1}\cup E_{t_3}$ we set $\lambda(e)=\{1,2\}$, obtaining a total of $2(\eta+\mu+1)+2(\mu+x+1)$ such labels. 
    For every vertex $u_i\in U$, we consider an arbitrary set $s_{u_i} \in \mathcal{C}^*$ that covers $u_i$, i.e. $u_i\in s_{u_i}$, and assign the labels $\lambda(e)=\{1,2\}$ to the corresponding edge $e=\{u_i,s_{u_i}\}$, accounting for a total of $2\eta$ such labels.    
    Finally, we assign the set of labels $\{1,2\}$ to all edges of the form $\{w_l,s_j\}$ with $w_l\in W$ and $s_j\in \mathcal{C}^*$, accounting for $2xk^*$ labels. Thus, in total we have that $|\lambda|=4(\eta+\mu+1)+2x+2xk^* \leq 6x+2xk^*$.
    
    It remains to prove that $(G,\lambda)$ is temporally connected. Denote by $G'$ the spanning subgraph of $G$ containing all the edges that receive the labels $\{1,2\}$ in $\lambda$. Observe that proving that $G'$ has diameter 2 is equivalent to proving that $(G,\lambda)$ is temporally connected, since every path $v_1,v_2,v_3$ of length 2 in $G'$ correspond to a valid temporal path $(\{v_1,v_2\},1),(\{v_2,v_3\},2)$ of length 2 in $(G,\lambda)$.

    Let $u_i\in U$. In graph $G'$, the vertex $u_i$ can reach every other vertex in at most 2 steps. To see this, observe that $u_i$ can reach vertex $v \in \mathcal{C}\cup U\cup \{t_1,t_2\}$ using the edge $\{u_i,t_1\}$ and then (if necessary) $\{t_1,v\}$. Moreover, $u_i$ can reach any vertex $v\in W\cup \{t_3\}$ by first moving into the arbitrarily chosen vertex $s_{u_i}$ in $\mathcal{C}^*$ (see the construction of $\lambda$) and then reaching $v$, thanks to the fact that all edges of the form  $\{s_{u_i},v\}$ with $v\in W\cup \{t_3\}$ belongs to $G'$.

    Let $s_j\in \mathcal{C}$. The vertex $s_j$ can reach every vertex $v\in \mathcal{C}\cup U\cup \{t_1,t_2\}$ with a path of length 2, by using the edges in $E_{t_1}$. Moreover, $s_j$ can reach every vertex $v\in W\cup \{t_3\}$ with a path of length 2 by using the edges in $E_{t_3}$.

    The vertex $t_1$ can reach every vertex $v\in \mathcal{C}\cup U \cup \{t_2\}$ using one edge in $E_{t_1}$. Moreover, $t_1$ can reach vertex $t_3$ using the path $\{t_1,t_2\},\{t_2,t_3\}$, and it can reach every vertex $w\in W$ with the path $\{t_1,s_j\},\{s_j,w\}$ where $s_j\in \mathcal{C}^*$.

    The vertex $t_2$ can reach all other vertices by first moving to either $t_1$ or $t_3$, and then using the edges in $E_{t_1}\cup E_{t_3}$.

    It remains to check the vertices $W\cup \{t_3\}$. Since all other vertices can reach these vertices with a path of length at most 2, we only need to check that $w\in W\cup \{t_3\}$ can reach $w' \in W\cup \{t_3\}$. Indeed, these vertices are connected with a path of length 2 via the edges in $E_{t_3}$.
\end{proof}

Now, we are ready to prove Theorem~\ref{thm:malinapprox}. For convenience of the reader, we restate the statement of Theorem~\ref{thm:malinapprox}.
\malinapprox*
\begin{proof}[Proof of Theorem~\ref{thm:malinapprox}]
    Let $(U,\mathcal{C})$ be an instance of SC with $\eta=|U|$ and $\mu = |C|$, let $\mathcal{C}^*$ be an optimal set cover, and let $k^*=|\mathcal{C}^*|$. Let $G$ be the graph constructed as previously described with $x=\eta+\mu+1$. Let $\lambda^*$ denote an optimal solution to MAL on instance $(G,2)$. Assume there is an $\alpha$-approximation algorithm for MAL and denote the solution given by such an algorithm with $\lambda_{\APX}$.     
    Denote with $G_p$ where $p\in [2]$ the subgraph of $G$ having only the edges that in $\lambda_{\APX}$ contain label $p$, that is $E(G_p)=\{e\in E(G) \mid p\in \lambda_{\APX}(e)\}$. Observe that in $G_p$ the existing edges going from $w_l$, $l\in[x]$, to the vertices in $\mathcal{C}$ must induce a set cover as the only paths of length 2 between $w_l$ and any vertex in $U$ must go through a vertex in $\mathcal{C}$.     
    More formally, the set of vertices $C_{p,l}=\{s_j\in \mathcal{C}\mid p \in \lambda_{\APX}(\{w_l,s_j\})\}$ is a set cover for every $p\in [2],l\in [x]$. Denote with $k$ the minimum among the size of sets $C_{p,l}$, that is $k=\min_{p\in [2],l\in [x]}|C_{p,l}|$.
    Observe that for each $p\in [2]$ and $l\in [x]$, $\lambda$ contains $|C_{p,l}|$ distinct labels to connect vertex $w_l$ to all vertices in $C_{p,l}$ with label $p$.
    Hence, $|\lambda_{\APX}| \geq \sum_{p\in[2]}\sum_{l\in[x]}|C_{p,l}| \geq 2xk$.

    By Lemma~\ref{lem:mal2}, we can compute a labeling $\lambda$ such that $|\lambda|\leq 6x+2xk^*$ and therefore, $|\lambda^*|\leq 6x+2xk^*$.
    
    As $\lambda_{\APX}$ is an $\alpha$-approximation algorithm for $(G,2)$, then $|\lambda_{\APX}| \leq \alpha |\lambda^*|$, and therefore $2xk\leq |\lambda_{\APX}|\leq \alpha |\lambda^*| \leq \alpha(6x +2xk^*)$. Dividing by $2x$, we get that $k\leq \alpha (k^*+3)$.

    Let $\alpha \leq \epsilon \log(|G|)$, where $|G| = O(N^2)$. We have $k\leq \epsilon'\log(N^2)k^* = 2\epsilon'\log(N)k^*$, for any $\epsilon'=\epsilon + o(1)$, which contradicts Theorem~\ref{lemma:sc} for any $\epsilon'\in (0,1/2)$. Hence, we have a contradiction for any $\alpha \leq \epsilon \log(|G|)$ and $\epsilon\in (0,1/2)$. The theorem follows by observing that $|G|=O(|V(G)|^2)$.   
\end{proof}

In the next theorem, we extend our logarithmic lower bound to the case where $a$ is any fixed value greater than or equal to 3. Specifically, we prove this result for \DCSS and then use Theorem~\ref{thm:apx-mal-dss} to show the same result for \MAL.

\begin{theorem}\label{thm:dcssinapprox}
For every $\epsilon\in(0,1/6)$, there is no polynomial-time $(\epsilon \log |V|)$-approximation algorithm for \DCSS, unless $\text{P} = \text{NP}$. The hardness holds even when the required diameter $d$ is a fixed parameter greater or equal to 3.
\end{theorem}

\begin{figure}[t]
    \centering
    \pgfdeclarelayer{nodelayer}
    \pgfdeclarelayer{edgelayer}
    \pgfsetlayers{nodelayer,edgelayer}

    \includegraphics[width=0.5 \textwidth]{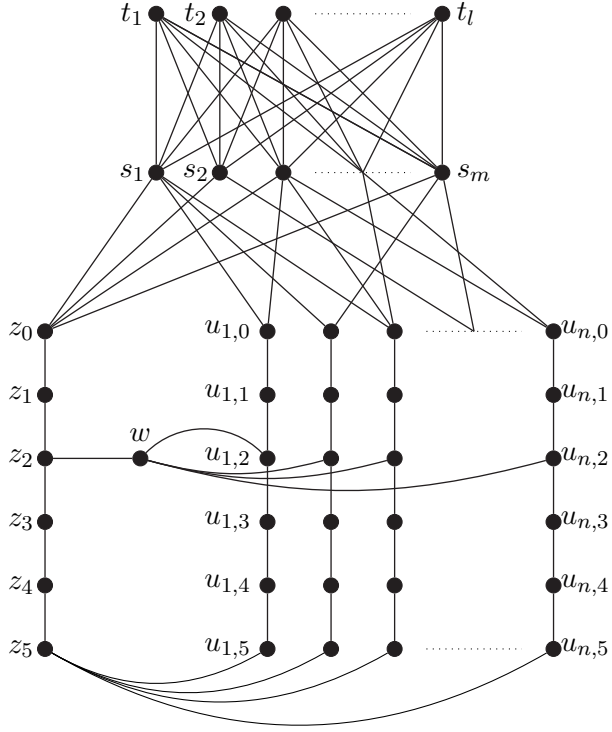}
    \caption{Example of the construction used in Theorem~\ref{thm:dcssinapprox} for $ d = 7 $.}
    \label{fig4}
\end{figure}

Let $(U,\mathcal{C})$ be an instance of SC. We construct an instance $(G=(V,E),d)$ of \DCSS for any fixed parameter $d\geq 3$. Fix $x\geq 3d\eta+\mu$ be fixed. We construct the graph $G=(V,E)$ as follows. See Figure~\ref{fig4} for a reference of the construction.

The set of vertices of $G$ consists of a vertex $s_i$ for each set in $\mathcal{C}$, a set of new vertices $T=\{t_1,t_2,\dots,t_x\}$, and $d-1$ copies of each vertex in $\Universe$. Specifically, for each $j\in[d-2]_0$, we define $\Universe_j=\{u_{i,j} \mid  u_i \in \Universe \}$. The vertices $u_{i,0}$ correspond to the elements of $\Universe$ and are connected to the sets in $\Sets$ according to the natural set-cover correspondence. Additionally, we introduce a set of new vertices $Z=\{z_0,z_1,\dots,z_{d-2}\}$ and one more new vertex, denoted $w$. Thus, the set of vertices of $G$ is $V(G) = T \cup \Sets \cup (\bigcup_{j\in[d-2]_0} \Universe_j) \cup Z \cup \{w\}$.

We now define the set of edges of $G$. Let $E_{\text{sc}}$ be the edges corresponding to the instance of set cover, that is,
\[
E_{\text{sc}} = \{(s_j, u_{i,0}) \mid s_j \in \mathcal{S} \land u_i \in \mathcal{U} \land u_i \in s_j \}.
\]
Next, we connect the vertices in $T$ to each of the vertices in $\mathcal{S}$ by defining
\[
E_t = \{(t_l, s_j) \mid s_j \in \mathcal{S} \land l \in [x]\}.
\]
We now add edges to form a path between the copies $u_{i,0}$ and $u_{i,d-2}$ of the same element $u_i \in \mathcal{U}$. Specifically,
\[
E_p = \{(u_{i,q}, u_{i,q+1}) \mid u_i \in \mathcal{U} \land q \in [d-3]_0\}.
\]
The vertex $z_{d-2}$ will be connected to all the vertices in $\mathcal{U}_{d-2}$, i.e.,
\[
E_{\text{zu}} = \{(z_{d-2}, u) \mid u \in \mathcal{U}_{d-2}\}.
\]
Similarly, vertex $z_0$ will be connected to all the vertices in $\mathcal{S}$, that is,
\[
E_{\text{zs}} = \{(z_0, s) \mid s \in \mathcal{S}\}.
\]
Moreover, $z_0$ and $z_{d-2}$ will be connected through a path that crosses the other vertices in $Z$, i.e.,
\[
E_{\text{zz}} = \{(z_q, z_{q+1}) \mid q \in [d-3]_0\}.
\]
Finally, we add the edges from vertex $w$ to all the vertices in $\mathcal{U}_{\lfloor \frac{d-2}{2} \rfloor} \cup \{z_{\lfloor \frac{d-2}{2} \rfloor}\}$, that is,
\[
E_w = \{(w, u) \mid u \in \mathcal{U}_{\lfloor \frac{d-2}{2} \rfloor}\} \cup \{(w, z_{\lfloor \frac{d-2}{2} \rfloor})\}.
\]
In conclusion, the edges of $G$ are
\[
E = E_{\text{sc}} \cup E_t \cup E_p \cup E_{\text{zu}} \cup E_{\text{zs}} \cup E_{\text{zz}} \cup E_w.
\]

We now prove that for any $d\geq 3$, if ($\Universe,\Sets$) admits a SET COVER of size $k^*$, then we can compute a spanner of $G$ with diameter $d$ and at most $\eta(d+1)+xk^*+\mu+d-2$ edges.

\begin{lemma}\label{lem:dss2}
    Let $(U,\mathcal{C})$ be an instance of SC and let $G$ be the graph constructed as described above. Given a set cover for $(U,\mathcal{C})$ of optimum size $k^*$, we can construct a feasible solution $H$ to the instance $(G,d)$ of \DCSS such that $|E(H)|\leq \eta(d+1)+xk^*+\mu+d-2$.
\end{lemma}
\begin{proof}
Let $\Sets^* = \{s^*_1,s^*_2,\dots,s^*_{k^*}\} \subseteq \Sets$ be an optimal set cover of size $k^*$. We construct a solution $H=(V,E_H)$ to the \DCSS problem in the following way. The set of vertices of $H$ is the same as that of $G$, while the set of edges will contain the sets $E_{w}$, $E_{p}$, $E_{zu}$, $E_{zp}$, and $E_{zs}$. $H$ will also contain all the edges of $E_t$ of the form $(t_l,s^*_j)$ for $l\in [x]$ with $s^*_j\in \Sets^*$. Finally, for every vertex $u_i\in \Universe_0$, we arbitrarily select from the set $E_{sc}$ one edge $(s^*_j,u_i)$  such that $s^*_j\in S^*$ and $u_i\in s^*_j$. Summing up the size of the chosen edges, we get that $|E_H|\leq \eta(d+1)+xk^*+\mu+d-2$.

It remains to prove that $H$ is connected and has diameter at most $d$. 

First, observe that the vertices in $T$ can reach, in at most $d$ steps, any vertex in $\bigcup\limits_{j\in[d-2]_0} \Universe_j \cup \{w\}$, since they are incident to a set cover. Moreover, they can also reach all vertices in $Z$ as well as the vertices $s\not \in S^*$ via a path passing through $z_{d-2}$. Therefore, we can restrict our attention to the vertices in $ \Sets \cup ( \bigcup\limits_{j\in[d-2]_0} \Universe_j ) \cup Z \cup \{w\}$.

Observe that $w$ can reach every other vertex through a path of length at most $\lfloor \frac{d}{2} \rfloor +1$. Moreover, a vertex in $\bigcup\limits_{j=0}^{\lfloor \frac{d-2}{2} \rfloor} \Universe_j \cup \{z_j\}$ can reach a vertex in $\Sets\cup (\bigcup\limits_{j=\lfloor \frac{d-2}{2} \rfloor +1}^{d-2} \Universe_j ) \cup \{z_j\}$ via a path traversing $w$, with length at most $d$. Therefore we may restrict our attention to the case where both vertices belong to one of the two aforementioned sets of vertices. 

Let $x$ and $y$ be any two vertices in $\bigcup\limits_{j=0}^{\lfloor \frac{d-2}{2} \rfloor} \Universe_j \cup \{z_j\}$. Observe that both $x$ and $y$ lie on a cycle of length at most $2*\lfloor \frac{d-2}{2} \rfloor +4$ that traverses vertices $z_{d-2}$ and $w$. Since we can upper-bound the length of the cycle by $d+3$, it follows that the distance between $x$ and $y$ is at most $\lfloor \frac{d+3}{2} \rfloor$, which is no greater than $d$ (using the fact that $d\geq 3$).

The only remaining case is when both $x$ and $y$ belong to $\Sets\cup (\bigcup\limits_{j=\lfloor \frac{d-2}{2} \rfloor +1}^{d-2} \Universe_j )\cup \{z_j\}$. In this case, both $x$ and $y$ lie on a common cycle traversing $w$ and $z_0$ with length at most $2 \lfloor \frac{d}{2}\rfloor +4 \leq d+4$. Therefore, the distance between $x$ and $y$ is at most $\lfloor \frac{d+4}{2} \rfloor$, which is at most $d$ as long as $d\geq 3$. This completes the analysis of all cases and concludes the proof.
\end{proof}

We can now prove Theorem~\ref{thm:dcssinapprox}.
\begin{proof}[Proof of~Theorem~\ref{thm:dcssinapprox}]
Let $d\geq 3$ be a fixed integer. Let ($\Universe,\Sets$) be an instance of SC, for which we can assume, without loss of generality, that $\eta+\mu\geq d$. Let $G$ be the graph constructed as described above with $x=\eta d+\mu$. Let $k^*$ be the optimum value for SC on ($\Universe,\Sets$). Let $H^*$ denote an optimal solution to \DCSS for the corresponding instance $(G,d)$. Assume there exists an $\alpha$-approximation algorithm for \DCSS and let $H_{\APX}$ denote the solution produced by this algorithm for $(G,d)$.

Observe that in any subgraph of $G$ with diameter at most $d$, the edges going from each $t_l\in T$ with $l\in [x]$ to the vertices in $\Sets$ must induce a set cover, as the only paths of length $d$ between $t_l$ and any vertex in $\Universe_{d-2}$ must traverse a vertex in $\Sets$. 

Let us denote by $C_l:=\{s \mid s\in \Sets \land (t_l,s)\in \DCSS_{\APX}\}$  and set $k=\min\limits_{l=1,\dots,x} |C_l|$. Observe that each element $s$ of $C_{l}$ corresponds to a distinct edge $\{t_l,s\}$. This implies that $|H_{\APX}|\geq\sum_{l\in[x]} |C_{l}|\geq kx$.

Moreover, $|H^*|\leq \eta (d+1)+\mu+k^*x+d-2$, since, by Lemma~\ref{lem:dss2}, we can construct a feasible solution to $(G,d)$ with the number of edges given on the right-hand side of the inequality.

Since we have that $|H_{\APX}|\leq \alpha|H^*|$, it follows that $$kx\leq \alpha(xk^*+\eta d+\mu +\eta+d-2),$$ which simplifies to
$$k\leq \alpha(k^*+\frac{\eta d+\mu}{x}+\frac{\eta +d-2}{x}).$$
As $x=\eta d+\mu$, we have that $k\leq \alpha(k^*+2)$.

Let $\alpha \leq \epsilon \log(|G|)$, where $|G| = O(dN^2)$. We have that $k\leq \epsilon'\log(dN^2)k^* \leq 3\epsilon'\log(N)k^*$, for any $\epsilon'=\epsilon + o(1)$, assuming that $d\leq N$. This is a contradiction to Theorem~\ref{lemma:sc} for any $\epsilon'\in (0,1/3)$. Therefore, we obtain a contradiction for any $\alpha \leq \epsilon \log(|G|)$, where $\epsilon\in (0,1/3)$. The statement follows by observing that $|G|=O(|V(G)|^2)$.
\end{proof}

Combining Theorem~\ref{thm:dcssinapprox} with Theorem~\ref{thm:apx-mal-dss}, we obtain the following corollary.

\begin{corollary}\label{cor:malinapprox}
  For every $\epsilon\in(0,1/(6a))$, there is no polynomial-time $(\epsilon \log |V|)$-approximation algorithm for \MAL, unless $\text{P}\!=\!\text{NP}$. The hardness holds even when the maximum allowed age $a$ is a fixed parameter greater or equal to 3.
\end{corollary}
\begin{proof}
    Let $(G,a)$ be an instance of \DCSS where $a=D_G$. Let $\epsilon \in (0,1/(6a))$ and assume that we have a $(\epsilon \log |V|)$-approximation algorithm for \MAL. By Theorem~\ref{thm:apx-mal-dss}, we know that there exists an $(a\epsilon \log |V|)$-approximation algorithm for \DCSS. Since $\epsilon\in (0,1/(6a))$, this contradicts Theorem~\ref{thm:dcssinapprox}.
\end{proof}

The next theorem presents our second lower bound on the approximation of \DCSS, based on a stronger complexity hypothesis. Again, we prove the hardness for \DCSS and then combine this result with Theorem~\ref{thm:apx-mal-dss} to obtain a hardness bound for \MAL.

\begin{restatable}{theorem}{stronginapprox}\label{thm:stronginapprox}
 For any constant $\epsilon\in (0,1)$, there is no polynomial-time $2^{\log ^{1-\epsilon} n}$-approximation algorithm for  \DCSS, unless $\text{NP}\subseteq \text{DTIME}(2^{\text{polylog}(n)})$.  The hardness holds even when the required diameter $d$ is a fixed parameter greater or equal to 3.
\end{restatable}

In order to prove Theorem~\ref{thm:stronginapprox}, we use a reduction from a problem called MIN-REP, introduced in~\cite{kortsarz}.
In MIN-REP, we are given a bipartite graph $\Tilde{G}=(A,B,\Tilde{E})$, where $A$ and $B$ are disjoint sets of vertices, and $\Tilde{E}$ represents the edges between vertices in $A$ and $B$. The set $A$ is partitioned into $r$ groups $A_1,A_2,\dots,A_r$ and the set $B$ is partitioned into $r$ groups $B_1,B_2,\dots,B_r$. Each set $A_i$ and $B_j$ has the same size, denoted by $\sigma$. 

The graph $\Tilde{G}$ induces a bipartite \emph{condensed graph} $G' = (U, W, E')$ defined as follows: 
\[
U = \{a_i \mid i \in [r]\}, \quad W = \{b_j \mid j \in [r]\}, 
\]
and $E'$ contains an edge between vertices $a_i \in U$ and $b_j \in W$ if and only if there is an edge in $\Tilde{E}$ between some vertex in $A_i$ and some vertex in $B_j$. 

The vertices in $V(G')$ are called \emph{condensed vertices}, and the edges in $E(G')$ are called \emph{condensed edges}.

A \emph{REP-cover} is a set $C\subseteq A \cup B$ of vertices with the property that it ``covers'' every condensed edge. Specifically, for all condensed edges $\{a_i,b_j\}$ there exist vertices $a\in C\cap A_i$ and $b\in C\cap B_j$ such that $\{a,b\}\in \Tilde{E}$. The objective of MIN-REP is to construct a REP-cover of minimum size.
Let us denote with $C^*$ an optimal solution to MIN-REP.

We say that an instance of MIN-REP is a YES-instance if $|C^*|=2r$ (one vertex for each group) and is a NO-instance if $|C^*| \geq 2^{\log ^{1-\epsilon'} n}r$. The following theorem is due to Kortsarz~\cite{kortsarz}. We use the statement given in \cite{CDKL} Theorem 6.1.

\begin{theorem}[\cite{kortsarz}]\label{thm:minrep}
    For any constant $\epsilon'\in (0,1)$, there is no polynomial-time algorithm that can distinguish between YES and NO instances of \text{MIN-REP}, unless $\text{NP}\subseteq \text{DTIME }(2^{\text{polylog}(n)})$.
\end{theorem}

Let $\Tilde{G}=(A,B,\Tilde{E})$ be an instance of MIN-REP with the associated condensed graph $G'=(U,W,E')$. We denote by $\Gamma(w)$ the group in $\Tilde{G}$ corresponding to the condensed vertex $w \in U \cup W$. Specifically, for $u \in U$, we have $\Gamma(u) = A_i$ for some $i \in [r]$, and for $w \in W$, we have $\Gamma(w) = B_j$ for some $j \in [r]$. Let $x$ be a parameter that we will set later.

Before we go into the details of the reduction from MIN-REP to \DCSS, we provide some intuition behind the construction. We start with a MIN-REP instance $\Tilde{G}=(A,B,\Tilde{E})$ and construct an instance $(G,3)$ of \DCSS. The graph $G$ contains all the vertices and edges of $\Tilde{G}$, along with a vertex for every group (corresponding to the condensed vertex), which is connected to all the vertices in its group. 

The idea is that in the \DCSS instance, we want two condensed vertices to be connected through the edges in $\Tilde{G}$ if they are connected by a condensed edge in $G'$. If they are not connected by a condensed edge, we want to have some other path between them. This ensures that the selected edges will correspond to a REP-cover. To achieve this, we add an edge between two condensed vertices if and only if they are not connected by a condensed edge in $G'$.

To connect all the vertices in $G$ that are not condensed vertices, we introduce a vertex $t$, which will act as the center of a star connecting all such vertices. This new vertex will be adjacent to all the vertices in $A$ and $B$, and it will be connected to the condensed vertices through disjoint paths of length two. In this way, the condensed vertices will be the only vertices that are not connected through a path of length at most 3 passing through the star.

The labels of edges incident to $t$ may introduce a dominant term in the size of the solution. To address this, we replicate the condensed vertices into $x + 1$ copies, where $x$ copies have the same role as described before (i.e., they try to reach the other condensed vertices through the MIN-REP instance, selecting their own vertices for the REP-cover), and the remaining copy has the role of connecting its copies to the other vertices.

We now construct our \DCSS instance. See Figure~\ref{fig_rid_3dpss} for an illustration. For simplicity, we first describe the reduction for the case of diameter 3, and then discuss how to extend it to any fixed $d\geq 3$.
\begin{figure}[t]
    \centering
    \pgfdeclarelayer{nodelayer}
    \pgfdeclarelayer{edgelayer}
    \pgfsetlayers{nodelayer,edgelayer}

    \includegraphics[width=0.8 \textwidth]{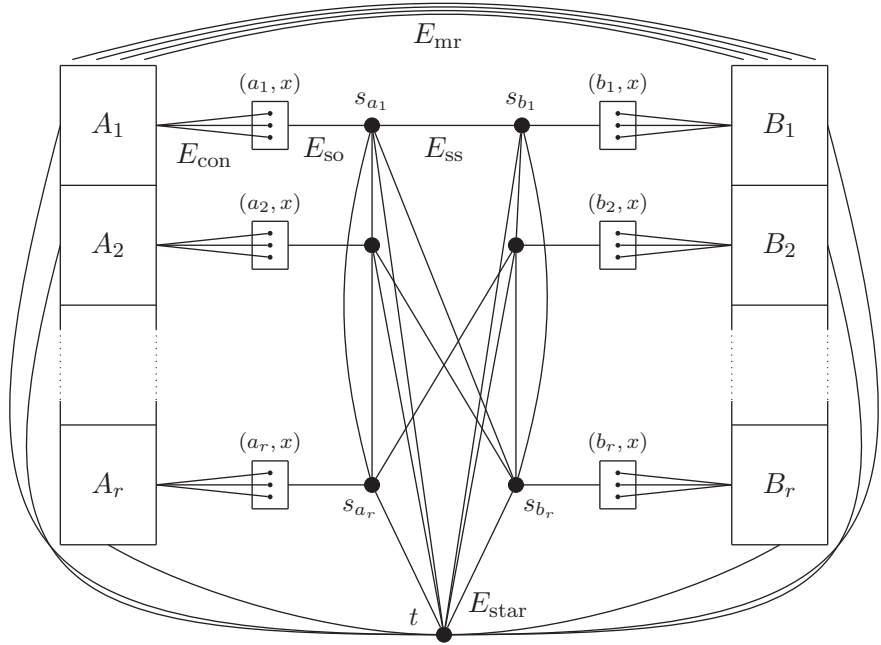}
    \caption{The graph used in the proof of Theorem~\ref{thm:stronginapprox}. Edges between a vertex and a set indicate that the vertex is connected to all vertices in the set.}
    \label{fig_rid_3dpss}
\end{figure}

We begin by defining the vertex set $V=A\cup B \cup V^L\cup V^R\cup S \cup  \{t\}$, where $V^L = U \times [x]$, $V^R = W \times [x]$, and $S=\{s_u : u\in U\cup W\}$. In other words, the vertices of $V^L$ and $V^R$ are $x$ copies of the condensed vertices,  $S$ contains one vertex for each condensed vertex, and we add a special vertex $t$ that will help us easily connect the majority of the vertices in the instance. 

Now, we define the edge set of $G$. We start with the edges of the MIN-REP instance, which we denote as $E_{\text{mr}} = \Tilde{E}$.
Next, we add the edges connecting the vertices of $V^L$ and $V^R$ to their corresponding group. More formally, let
$ E_{\text{con}} = \{ \{(u,i),a \} \mid u \in U \cup W, \, a \in \Gamma(u), \, i \in [x] \}.$
We then connect the vertices in $S$ using two edge sets. The first set connects each $s_u \in S$ to the vertices $(u,i) \in V^L \cup V^R$, and is defined as
$ E_{\text{so}} = \{ \{(u,i),s_u \} \mid u \in U \cup W, \, i \in [x] \}. $
The second set connects the vertices within $S$ such that there is an edge between two vertices in $S$ if and only if there is no condensed edge connecting the corresponding condensed vertices in $G'$. Formally, this set is
$ E_{\text{ss}} = \{ \{s_u,s_{u'}\} \mid u,u' \in U \cup W, \, \{u,u'\} \not\in E(G') \}. $

Observe that for two vertices in $U$, the corresponding vertices in $S$ will be connected by an edge, and the same holds for the vertices in $W$.

Finally, we add the edges that have $t$ as an endpoint. These edges form the set
$ E_{\text{star}} = \{ \{t,y\} \mid y \in A \cup B \cup S \}. $

The final edge set of $G$ is the union of all the previous edge sets:
$ E(G) = E_{\text{mr}} \cup E_{\text{con}} \cup E_{\text{so}} \cup E_{\text{ss}} \cup E_{\text{star}}. $

The next lemma shows that the diameter of $G$ is equal to 3. 
\begin{lemma}\label{lemma:diameterthree}
The graph $G$ constructed above has diameter 3.
\end{lemma}
\begin{proof}
First, observe that vertex $t$ can reach every vertex in $A \cup B \cup S$ with one edge, and it can reach every vertex in $V^L \cup V^R$ with a path of two edges going through $S$. Let $v_1, v_2$ be two vertices of $V(G) \setminus \{t\}$. If at most one of $v_1$ and $v_2$ belongs to $V^L \cup V^R$, then there exists a path between $v_1$ and $v_2$ of length at most 3 going through vertex $t$. 

The only remaining case is when both $v_1$ and $v_2$ belong to $V^L \cup V^R$. Therefore, $v_1$ and $v_2$ are of the form $v_1 = (w, i)$ and $v_2 = (w', j)$ with $w, w' \in U \cup W$. If $w = w'$, then $v_1$ and $v_2$ can reach each other in 2 steps using the edges incident to $s_w$. If both $w$ and $w'$ lie in $U$ (analogously in $W$), then we can use the edge $\{s_w, s_{w'}\} \in E_{\text{ss}}$ to connect $v_1$ to $v_2$ with a path of length 3. 

Finally, we have the case where $w \in U$ and $w' \in W$ (or analogously $w \in W$ and $w' \in U$). If there exists an edge between $s_w$ and $s_{w'}$, then the path of length 3 is $(w, i), s_w, s_{w'}, (w', j)$. Otherwise, if there is no edge between $s_w$ and $s_{w'}$, then, by the construction of $G$, there exists a superedge between $w$ and $w'$ in $G'$, and therefore there exists an edge $\{a, b\} \in \Tilde{E} = E_{\text{mr}}$ with $a \in \Gamma(w)$ and $b \in \Gamma(w')$. In conclusion, the path of length 3 between $v_1$ and $v_2$ in $G$ is $(w, i), a, b, (w', j)$, and therefore $D_G = 3$.
\end{proof}

We are going to show that if there exists a $2^{\log ^{1-\epsilon} n}$-approximation algorithm for an instance $(G,3)$ of \DCSS, then Theorem \ref{thm:minrep} is contradicted. To this aim, we need two lemmas.

Let $H$ be an arbitrary spanning subgraph of $G$ with $D_H=3$. The first lemma proves that the size of $E(H)$ is lower bounded by $x$ times the size of an optimal solution to the MIN-REP instance, i.e., $|E(H)| > x \cdot |C^*|$, where $x$ is the number of copies in $G$ of each condensed vertex.

\begin{lemma}\label{lemma:sound}
    For any feasible solution $H$ to the instance $(G,3)$ of \DCSS we have $|E(H)|> x \cdot |C^*|$.
\end{lemma}
\begin{proof}
    For $i \in [x]$, let $C_i \subseteq A \cup B$ be the set of vertices in $A \cup B$ that are adjacent to any $i$-th copy of a supervertex in $H$. Formally, 
    $C_i = \{ c \in A \cup B \mid w \in U \cup W \land c \in \Gamma(w) \land \{(w,i), c\} \in H \}.$
    We will prove that $C_i$ is a REP-cover of $\Tilde{G}$ for all $i \in [x]$. Fix an $i \in [x]$ and a superedge $\{u,v\} \in E(G')$. Observe that $\{u,v\}$ does not belong to $G$ and that the only possible paths between $(u,i)$ and $(v,i)$ in $H$ of length at most 3 are those passing through an edge in $E_{\text{mr}}$. Hence, there must exist two vertices $a$ and $b$ such that $a \in \Gamma(u)$, $b \in \Gamma(v)$, and the path $(u,i),a,b,(v,i)$ is in $H$. By definition of $C_i$, both $a$ and $b$ belong to $C_i$, and $C_i$ covers the superedge $\{u,v\}$. Hence, for all $i \in [x]$, we have that $C_i$ is a valid REP-cover, and consequently $|C_i| \geq |C^*|$. 

    Observe that every vertex $c$ belonging to $C_i$ corresponds to a distinct edge $\{c, (u,i)\}$ in $H$, for some $u \in U \cup W$. Moreover, for each $u \in U \cup W$, and $i,j \in [x]$ with $i \neq j$, the set of edges of $G$ incident to $(u,i)$ and $(u,j)$ are disjoint. Therefore, the number of edges in $H$ is at least the sum, over all $i \in [x]$, of the size of $C_i$, that is, 
    $|E(H)| > \sum_{i=1}^{x} |C_i|$, where the inequality is strict as we need to add more edges in order to connect the other vertices (e.g., $t$). Finally, recalling that $|C_i| \geq |C^*|$, we conclude that $|E(H)| > x \cdot |C^*|$.
\end{proof}

The second lemma shows that there exists a spanning subgraph of diameter 3 which does not cost much more than $C^*$ in a YES-instance of MIN-REP.

\begin{lemma}\label{lemma:complete}
    Assume that $\Tilde{G}$ is a YES-instance of MIN-REP, then there exists a feasible solution $H$ to the instance $(G,3)$ of \DCSS such that $|E(H)|\leq 4rx + 5r^2\sigma$.
\end{lemma}
\begin{proof}
    Let $\Tilde{G}$ be a YES-instance of MIN-REP, so there exists a REP-cover $C$ of size $2r$. Since $|C| = 2r$, it follows that $|C \cap \Gamma(w)| = 1$ for each $w \in U \cup W$. Let $c_w$ denote the unique element of $C \cap \Gamma(w)$ for each $w \in U \cup W$. 
    
    We construct $H$ starting from a subgraph of $G$ that includes the edges $E_{\text{star}} \cup E_{\text{ss}} \cup E_{\text{so}}$. We then add the subset of edges in $E_{\text{con}}$ of the form $\{(w,i), c_w\}$ for all $w \in U \cup W$ and $i \in [x]$. Finally, we add the edges $\{a,b\} \in E_{\text{mr}}$ such that $a, b \in C$.
    
    Let us now analyze the size of $H$. The edges in $E_{\text{star}}$ are $2r + 2r\sigma$, the edges in $E_{\text{so}}$ contribute $2rx$, and we select $2rx$ edges from $E_{\text{con}}$. Finally, $H$ will include $\frac{2r(2r-1)}{2}$ edges in $E_{\text{ss}} \cup E_{\text{mr}}$. To see this, observe that for each vertex $s_w \in S$ with $w \in U \cup W$, we include $(r-1) + (r - \deg_{G'}(w))$ edges from $E_{\text{ss}}$ plus $\deg_{G'}(w)$ edges incident to $c_w$ in $E_{\text{mr}}$ (one for each superedge to be covered). 
    
    Therefore, in total, $H$ will contain at most $4rx + 2r^2 + r + 2r\sigma$ edges. Clearly, we can upper bound the total number of edges in $H$ by $4rx + 5r^2\sigma$.

    It remains to prove that $H$ has diameter at most 3. The proof is essentially the same as that of Lemma~\ref{lemma:diameterthree}.
\end{proof}

We are now ready to prove Theorem~\ref{thm:stronginapprox}. For convenience of the reader, we restate the statement of Theorem~\ref{thm:stronginapprox}.

\stronginapprox*

\begin{proof}
    Fix $\epsilon \in (0,1)$. Let $\Tilde{G}$ be an instance of MIN-REP of size $\Tilde{n}$ such that $\Tilde{n} \geq 2^{2^{\frac{2}{\epsilon(1-\epsilon)}}}$. 
    Let $x = r\sigma$. Clearly, the graph $G$ created as described earlier has size $n$, which is polynomial in $\Tilde{n} = 2r\sigma$. Specifically, we have $n \leq (2r\sigma)^2$.
    
    Assume there exists a polynomial-time $(2^{\log^{1-\epsilon} n})$-approximation algorithm for the \DCSS problem, which we apply to the instance $(G,3)$ previously defined. Let $H_{\text{APX}}$ denote the approximate solution, and let $H^*$ denote an optimal solution to the \DCSS instance on $(G,3)$.
    
    Theorem~\ref{thm:minrep} implies that, under the assumption that $\text{NP} \not\subseteq \text{DTIME}(2^{\text{polylog}(n)})$, there is no polynomial-time algorithm that can distinguish between the YES case, where $|C^*| = 2r$, and the NO case, where $|C^*| \geq r 2^{\log^{1-\epsilon'} (2r\sigma)}$. Since Theorem~\ref{thm:minrep} holds for any choice of $\epsilon'$, we fix $\epsilon' = \epsilon^2$, and we have $|C^*| \geq r 2^{\log^{1-\epsilon^2} (2r\sigma)}$ when $\Tilde{G}$ is a NO-instance for MIN-REP.
    
    By Lemma~\ref{lemma:complete}, if $\Tilde{G}$ is a YES-instance of MIN-REP, then $|H^*| \leq 9xr$. Therefore, when $\Tilde{G}$ is a YES-instance, the solution of our approximation algorithm on $G$ has size:
    \[
    |H_{\text{APX}}| \leq 9xr \cdot 2^{\log^{1-\epsilon} n} \leq 9xr \cdot 2^{2\log^{1-\epsilon} (2r\sigma)}.
    \]
    
    On the other hand, when $\Tilde{G}$ is a NO-instance of MIN-REP, by Lemma~\ref{lemma:sound}, we have 
    \[ 
    |H^*| > x \cdot |C^*| \geq xr 2^{\log^{1-\epsilon^2} (2r\sigma)} = xr 2^{(\log^{1-\epsilon} (2r\sigma))^{1+\epsilon}}.
    \]
    
    Now, using the fact that $2r\sigma = n' \geq 2^{2^{\frac{2}{\epsilon(1-\epsilon)}}}$ and $0 < \epsilon < 1$, we get the inequality:
    \[
    9xr 2^{2\log^{1-\epsilon} (n')} < xr 2^{(\log^{1-\epsilon} (n'))^{1+\epsilon}},
    \]
    which implies that we can use our approximation algorithm to distinguish between YES and NO instances of MIN-REP, contradicting Theorem~\ref{thm:minrep}. Hence, for any $\epsilon$, no polynomial-time $(2^{\log^{1-\epsilon} n})$-approximation algorithm can exist, and the theorem follows for $d = 3$.
    
    To extend the hardness to the case where the required diameter $d$ is any fixed value $d \geq 3$, we need to adjust the reduction slightly. For simplicity, assume that $d = 2k + 1$ for some integer $k \geq 2$. We change the reduction by adding, for each vertex $(u,i) \in V^L \cup V^R$, a path of length $k-1$ with copies of vertex $(u,i)$. Specifically, we add vertices $(u,i,1), (u,i,2), \dots, (u,i,k-1)$ and edges $\{(u,i), (u,i,1)\}, \{(u,i,j), (u,i,j+1)\}$ for $j \in [k-2]$. In this graph, vertices $(u,i,k-1)$ must select a REP-cover in order to connect to every other vertex with a path of length at most $d$. 
    
    As before, by choosing $\epsilon$ small enough and an instance of MIN-REP large enough, we can prove that $d$-\DCSS cannot be approximated with a factor better than $2^{\log^{1-\epsilon} n}$ for any odd $d \geq 3$. A similar construction can be done when $d$ is even, though more technicalities arise due to the different path lengths based on whether $(u,i) \in V^L$ or $(u,i) \in V^R$.    
\end{proof}

We now combine Theorem~\ref{thm:stronginapprox} with Theorem~\ref{thm:apx-mal-dss} to obtain a hardness lower bound for \MAL. We observe that in this result, differently from Corollary~\ref{cor:malinapprox}, the interval for $\epsilon$ does not depend on $a$ as $\epsilon$ appears in the exponent.
\begin{corollary}\label{cor:stronginapproxmal}
 For any constant $\epsilon\in (0,1)$, there is no polynomial-time $2^{\log ^{1-\epsilon} n}$-approximation algorithm for \MAL, unless $\text{NP}\subseteq \text{DTIME}(2^{\text{polylog}(n)})$. The hardness holds even when the maximum allowed age $a$ is a fixed parameter greater or equal to 3.

\end{corollary}
\begin{proof}
We prove that the corollary holds for any constant $\epsilon \in (0,1)$. To this end, assume by way of contradiction that there exists $\epsilon \in (0,1)$ such that there is a $2^{\log^{1-\epsilon} n}$-approximation algorithm for \MAL. By Theorem~\ref{thm:apx-mal-dss}, this implies the existence of a $(a \cdot 2^{\log^{1-\epsilon} n})$-approximation algorithm for \DCSS.
Observe that
$$a \cdot 2^{\log^{1-\epsilon} n} = 2^{\log^{1-\epsilon} n + \log a}.$$
We aim to find $\epsilon' \in (0,1)$ such that
$$2^{\log^{\epsilon'} n} > 2^{\log^{1-\epsilon} n + \log a},$$
which would contradict Theorem~\ref{thm:stronginapprox} as $ a \cdot 2^{\log^{1-\epsilon} n} < 2^{\log^{\epsilon'} n}$. Thus, it suffices to show that there exists $\epsilon'$ such that
$$\log^{\epsilon'} n >  \log^{1-\epsilon} n + \log a.$$
As $a$ is a fixed parameter, it suffices to pick a sufficiently large $\epsilon'\in (1-\epsilon,1)$.
\end{proof}

\section{Approximation algorithms for \DCSS}\label{sec:dcss}
In this section, we present approximation algorithms for \DCSS by exploiting related problems from the literature. In the following section, we show how to adapt the algorithms for \DCSS to \MAL by using Theorem~\ref{thm:apx-mal-dss}.

To the best of our knowledge, the \DCSS problem first appeared, in a slightly different form, only in~\cite{ElkinPeleg}. In particular, the version considered by~\cite{ElkinPeleg}, called the $d$-DSS problem, takes in input only a graph and asks for the minimum-size spanning subgraph of diameter at most $d$, where $d$ is fixed in the definition of the problem. In~\cite{ElkinPeleg} the results for $d$-DSS are obtained through a reduction from the All-Client $k$-spanner ($k$-AC) problem contained in an unpublished manuscript.
In $k$-AC, we are given a graph and a subset of the edges called servers, and we want to find a subgraph containing only server edges, in which the distance between any pair of vertices is at most $k$ times their original distance and the number of edges is minimized. 
The authors of \cite{ElkinPeleg}, show that $k$-AC can be approximated within a factor $O(\sqrt{n})$ with an additive violation on the required distances of at most 2 edges, and it is hard to approximate within a factor $2^{\log ^{1-\epsilon} n}$, unless $\text{NP}\subseteq \text{DTIME}(2^{\text{polylog}(n)})$. In the previous section, we presented a different reduction to derive this last result.

There are few other known optimization problems that generalize \DCSS, and thus, we can apply the approximation algorithms developed for these problems to \DCSS as well.

\paragraph*{Mincost Diameter-$d$ problem.}
Dodis and Khanna~\cite{DodisKhanna}, studied a generalization of the \DCSS problem called \emph{Mincost Diameter-$d$} in which we are given a graph $G=(V,E)$ with a cost function $c\xspace:\xspace \overline{E} \rightarrow \mathbb{R}^+$, where $\overline{E}=\{\{u,v\}\xspace:\xspace\{u,v\}\not \in E\}$, and we want to find a set of edges $E'\subseteq \overline{E}$ of minimum cost such that $G'=(V,E\cup E')$ has diameter at most $d$. There is a correspondence between \DCSS and Mincost Diameter-$d$ when the initial graph $G$ provided as input to the latter is empty. More precisely, let $(G,d)$ be an instance of \DCSS, we construct graph $G'=(V(G),\emptyset)$ and cost function $c$ such that $c(e)=1$ if $e\in E(G)$ and $c(e)=m+1$ otherwise. Observe that any feasible solution of cost at most $m$ for Mincost Diameter-$d$ problem on instance $(G',c)$ corresponds to a feasible solution for \DCSS on instance $(G,d)$ with the same cost, and vice versa. Moreover, we can disregard solutions with cost greater than $m$, since selecting all edges $e$ such that $c(e)=1$ provides a feasible solution to the Mincost Diameter-$d$ (recall that $d\geq D_G$) with strictly smaller cost.

Dodis and Khanna~\cite{DodisKhanna} proved that Mincost Diameter-$d$ is hard to approximate within $O(2^{\log^{1-\epsilon} n})$ unless $ \text{NP} \subseteq \text{DTIME}(n^{\text{poly log}(n)})$. However, the reduction used to establish this hardness result relies on an initial non-empty subgraph, and it does not appear to be trivially adaptable to the \DCSS problem. On the approximation side, the authors also proved the following result, which is of particular interest to us.

\begin{theorem}[\cite{DodisKhanna}, Section 5]
    There exists a $O(\log n)$-approximation algorithm for Mincost Diameter-2.
\end{theorem}
As observed earlier, the Mincost Diameter-$d$ problem can be viewed as a generalization of \DCSS, and therefore, the following result is implied by the previous one.
\begin{corollary}\label{cor:dodiskhanna}
    There exists a $O(\log n)$-approximation algorithm for \DCSS when $d=D_G=2$.
\end{corollary}

\paragraph*{Pairwise Spanner problem.}
In~\cite{CDKL}, the authors considered a problem called \emph{Pairwise Spanner}, which generalizes many other network design problems, including \DCSS. In the Pairwise Spanner problem, we are given a (possibly directed) graph $ G $, a collection of ordered vertex pairs $ P \subseteq V \times V $, and a function $ \Delta : P \to \mathbb{N} $, and we aim to compute a subgraph $ H $ of $ G $ that minimizes $ |E(H)| $, such that $ d_H(s,t) \leq \Delta(s,t) $ for all pairs $ (s,t) \in P $. In the same paper, the authors proved the following result regarding Pairwise Spanner.

\begin{theorem}[\cite{CDKL}, Theorem 1.8] \label{thm:pairwisespanner}
    For any constant $ \epsilon > 0 $, there exists a $ O(n^{3/5+\epsilon}) $-approximation algorithm for the \text{Pairwise Spanner} problem.
\end{theorem}

Observe that if we set $ P $ to be all ordered pairs of distinct vertices in $ G $ and fix $ \Delta(s,t) = d $ for all pairs $ (s,t) \in P $, we obtain the \DCSS problem with parameter $ d $, and hence the above approximation also applies to this case.

\begin{corollary}\label{cor:ps-approx-alg}
    For any constant $ \epsilon > 0 $, there exists a $ O(n^{3/5+\epsilon}) $-approximation algorithm for the \DCSS problem.
\end{corollary}

\paragraph*{Additive $+b$-spanner problem.}
An \emph{additive $+b$-spanner} (or simply $+b$-spanner) is a spanning subgraph that preserves all distances between vertices up to an additive stretch. More formally, a spanning subgraph $H\subseteq G$ is an additive $+b$-spanner  if $d_H(u,v)\leq d_G(u,v)+b$ for all $u,v\in V$. Additive spanners are widely studied in the literature; see the survey~\cite{surveyspanners} and the references therein. The following theorem summarizes several results from the literature on additive spanners.
\begin{theorem}[\cite{surveyspanners}, Section 8]
Let $G$ be any graph with $n$ vertices,
    \begin{enumerate}[i)]
        \item \cite{span-2}, $G$ admits a +2-spanner with $O(n^{3/2})$ edges.
        \item \cite{span-4}, $G$ admits a +4-spanner with $O(n^{7/5})$ edges.
        \item \cite{span-6}, $G$ admits a +6-spanner with $O(n^{4/3})$ edges.
    \end{enumerate}
\end{theorem}
Observe that, by definition, a $+b$-spanner of a graph $G$ is a feasible solution to \DCSS on input $(G,d)$, when $d\geq D_G+b$. Therefore, the previous theorem, along with the observation that an optimal solution for \DCSS has size at least $n-1$, implies the following corollary.

\begin{corollary}\label{cor:addspanner}
   Let $(G,d)$ be an instance of \DCSS.
     \begin{enumerate}[i)]
     \item If $d\geq D_G+2$ then there exists an $O(n^{1/2})$-approximation algorithm for \DCSS.
     \item If $d\geq D_G+4$ then there exists an $O(n^{2/5})$-approximation algorithm for \DCSS.
     \item If $d\geq D_G+6$ then there exists an $O(n^{1/3})$-approximation algorithm for \DCSS.
     \end{enumerate} 
\end{corollary}

\section{Approximation Algorithms for \MAL}\label{sec:alg}

In this section we introduce our approximation algorithms for \MAL. We recall that $D_G\leq a\leq 2D_G +2$.
The following theorem provides our first set of results for the cases in which $a$ is sufficiently large.
\begin{theorem}\label{thm:malradiusdiameter}
    Let $ (G, a) $ denote an instance of \MAL.
    \begin{enumerate}[i)]
        \item If $ a \geq 2R_G $, then we can compute a solution $ \lambda $ for \MAL in polynomial time such that $ |\lambda| \leq |\lambda^*| + 2 $.
        \item If $ a \geq 2R_G + 1 $, then we can compute a solution $ \lambda $ for \MAL in polynomial time such that $ |\lambda| \leq |\lambda^*| + 1 $.
        \item If $ a \geq 2D_G + 2 $, then we can compute a solution $ \lambda $ for \MAL in polynomial time such that $ |\lambda| = |\lambda^*| $.
    \end{enumerate}
\end{theorem}

\begin{proof}
Let $r$ be the center of $G$. We observe that $|\lambda^*|\geq 2|V(G)|-4$ since at least $2|V(G)|-4$ labels are necessary for a temporal graph to be temporally connected, even without any constraints on the lifetime (see~\cite{KMMS}). Moreover, if $G$ does not contain $C_4$, then $|\lambda^*|\geq 2|V(G)|-3$.

\emph{i}) We use the following folklore algorithm (see, e.g.,~\cite{KMMS}) that assigns 2 labels to the edges of a shortest path tree $ T $ rooted at the center $ r $ of graph $ G $. Each edge $ e = \{ u, v \} $ of $ T $, where $ d_G(u, r) = d_G(v, r) + 1 $, takes 2 labels: $ R_G - d_G(u, r) + 1 $ and $ R_G + d_G(u, r) $. The largest label used is therefore $ 2R_G $, and hence $ L_\lambda = 2R_G $.
Moreover, each pair of vertices $ (x, y) $ is temporally connected in $ (G, \lambda) $ as there is a temporal path from $ x $ to $ y $ passing through the least common ancestor $ lst $ of $ x $ and $ y $ in $ T $. In fact, let us consider the unique path $ P $ connecting $ x $ to $ y $ in $ T $. If $ x = lst $, then $ P $ has labels $ R_G + d_G(x, r) + 1, R_G + d_G(x, r) + 2, \ldots, R_G + d_G(y, r) $, which are strictly increasing as $ d_G(y, r) \geq d_G(x, r) + 1 $. Analogously, if $ y = lst $, then $ P $ has labels $ R_G - d_G(y, r) + 1, R_G - d_G(y, r) + 2, \ldots, R_G - d_G(x, r) $, which are strictly increasing as $ d_G(x, r) \geq d_G(y, r) $.
If $ lst \neq x, y $, let $ lst_x $ and $ lst_y $ be the two vertices adjacent to $ lst $ in $ P $ that are closer to $ x $ and $ y $, respectively. Then $ P $ has labels $ R_G - d_G(x, r) + 1, \ldots, R_G - d_G(lst_x, r) + 1, R_G + d_G(lst_y, r), \ldots, R_G + d_G(y, r) $, which are strictly increasing as $ d_G(x, r) \geq d_G(lst_x, r) $, $ d_G(lst_x, r) = d_G(lst_y, r) \geq 1 $, and $ d_G(lst_y, r) \leq d_G(y, r) $.
Finally, since each edge of $ T $ has exactly 2 labels, we have $ |\lambda| = 2E(T) = 2|V(G)| - 2 $, and hence $ |\lambda| \leq |\lambda^*| + 2 $.

\emph{ii}) We slightly modify the above algorithm as follows. Let $ r' $ be a vertex adjacent to $ r $ in a shortest path tree $ T $ rooted at $ r $, with $ r $ being the center of $ G $. The edge between $ r $ and $ r' $ takes a single label $ R_G + 1 $. Each edge $ e = \{ u, v \} $ in $ E(T) \setminus \{ \{ r, r' \} \} $, where $ d_G(u, r) = d_G(v, r) + 1 $, takes 2 labels: $ R_G - d_G(u, r) + 1 $ and $ R_G + d_G(u, r) + 1 $. Note that the labels in $ E(T) \setminus \{ \{ r, r' \} \} $ are similar to those used in the proof of (\emph{i}), with the only difference being that the second label is shifted by one time unit. Therefore, the largest label used is $ 2R_G + 1 $, and hence $ L_\lambda = 2R_G + 1 $. Moreover, the proof that each pair of vertices $ (x, y) $ is temporally connected is similar to the previous case (taking into account this difference). Finally, since $ E(T) \setminus \{ \{ r, r' \} \} = n - 2 $, the total number of labels used is $ |\lambda| = 2(|V(G)| - 2) + 1 = 2|V(G)| - 3 \leq |\lambda^*| + 1 $.

\emph{iii}) We distinguish between two cases: whether $ G $ contains $ C_4 $ or not. In the former case, we observe that the algorithm given in~\cite{KMMS} to temporally connect a graph containing $ C_4 $ assigns an optimal number of $ 2|V(G)| - 4 $ labels, where the largest label is no more than $ 2D_G + 2 $. In the latter case, the statement follows by observing that $ 2D_G + 2 \geq 2R_G + 1 $, and hence the labeling used to prove (\emph{ii}), which uses $ 2|V(G)| - 3 $ labels, is optimal.
\end{proof}

The next corollary collects the approximation results for the cases in which $a$ is slightly larger than $D_G$, which are derived from the relation between \MAL and \DCSS established in Theorem~\ref{thm:apx-mal-dss}.
\begin{corollary}\label{cor:approxmal-via-dss}
    Let $(G,a)$ denote an instance of \MAL.
     \begin{enumerate}[i)]
     \item If $a=2$, then there exists an $O(\log(n))$-approximation algorithm for \MAL.
     \item For any $\epsilon>0$, there exists an $O(D_G\cdot n^{3/5+\epsilon})$-approximation algorithm for \MAL.
     \item If $a\geq D_G+2$, then there exists an $O(D_G\cdot n^{1/2})$-approximation algorithm for \MAL.
     \item If $a\geq D_G+4$, then there exists an $O(D_G\cdot n^{2/5})$-approximation algorithm for \MAL.
     \item If $a\geq D_G+6$, then there exists an $O(D_G\cdot n^{1/3})$-approximation algorithm for \MAL.
     \end{enumerate}
\end{corollary}
\begin{proof}
i) By combining Corollary~\ref{cor:dodiskhanna} with Theorem~\ref{thm:apx-mal-dss}, we obtain an $O(a \log(n))$-approximation algorithm for \MAL. The statement follows by observing that $a$ is a constant.

ii) By combining Corollary~\ref{cor:ps-approx-alg} with Theorem~\ref{thm:apx-mal-dss}, we obtain an $O(a n^{3/5+\epsilon})$-approximation algorithm for \MAL. The statement follows from the observation that $D_G\leq a \leq 2D_G+2$. 

iii), iv), v) By combining Corollary~\ref{cor:addspanner} with Theorem~\ref{thm:apx-mal-dss}, we obtain approximation factors of $O(a n^{1/2})$, $O(a n^{2/5})$, and $O(a n^{1/3})$, respectively, for \MAL. The claim again follows from the observation that $D_G \leq a \leq 2D_G + 2$.
\end{proof}

Observe that there exists a trivial $O(n)$-approximation algorithm which consists of computing, for each node $s$, a shortest path tree rooted at $s$ and assigning label $d_G(s,v)$ to each edge $\{u,v\}$ in the tree, assuming that $d_G(s,v)=d_G(s,u)+1$. The bounds in Corollary~\ref{cor:approxmal-via-dss} outperforms this trivial bound only when $D_G \cdot n^\alpha \in o(n^2)$. In particular, the bounds in ii)--v) outperform the trivial bound when $D_G\in o(n^{2/5})$, $D_G\in o(n^{1/2})$, $D_G\in o(n^{3/5})$, and $D_G\in o(n^{2/3})$, respectively.


In the next theorem, we present our main algorithmic results, which consist of approximation algorithms for \MAL that do not exploit the relationship with \DCSS and whose approximation ratios do not depend linearly on $ D_G $.

\begin{restatable}{theorem}{malapprox}\label{thm:mal-approx-constr-a}
    Let $ (G, a) $ be an instance of \MAL.
    \begin{enumerate}[i)]
        \item If $ a \geq \lceil 3/2 \cdot D_G \rceil $, then there exists an $ O(\sqrt{n \log n}) $-approximation algorithm for \MAL.
        \item If $ a \geq \lceil 5/3 \cdot D_G \rceil $, then there exists an $ O(\sqrt[3]{D_G n \log^2 n}) $-approximation algorithm for \MAL.
    \end{enumerate}
\end{restatable}

Differently from Corollary~\ref{cor:approxmal-via-dss}, the bounds provided in Theorem~\ref{thm:mal-approx-constr-a} always outperform the trivial bound, regardless of the value of $D_G$. Moreover, the approximation guarantee of Theorem~\ref{thm:mal-approx-constr-a}~(i) matches or outperforms those of Corollary~\ref{cor:approxmal-via-dss}~(ii)-(v) when $ D_G \in \Omega(\sqrt{n^{1/3} \log n}) $, while Theorem~\ref{thm:mal-approx-constr-a}~(ii) achieves the same when $ D_G \in \Omega(\log n) $.

In order to prove Theorem~\ref{thm:mal-approx-constr-a}, we need to introduce some new concepts. 

The following definition is an adaptation to undirected graphs of Definition 2.2 from~\cite{ChoudharyGold}.

\begin{definition}\label{def:dsp}
    For a graph $G=(V,E)$ and a set-pair $(S_1,S_2)$, where $S_1,S_2\subseteq V$, we say that $(S_1,S_2)$ is a $\langle h_1,h_2\rangle$-dominating set-pair of size-bound $\langle n_1,n_2\rangle$ if $|S_1|=O(n_1)$, $|S_2|=O(n_2)$, and one of the following conditions holds.
    \begin{enumerate}
        \item For each $v\in V(G)$, $d_G(v,S_1)\leq h_1$,
        \item For each $v\in V(G)$, $d_G(v,S_2)\leq h_2$.
    \end{enumerate}
\end{definition}

Moreover, we need the following lemma regarding the well-known \emph{Hitting Set} (HS) problem. In the HS problem, we are given a universe $ U = \{u_1, \dots, u_n\} $ and a collection of $ N $ subsets $ \mathcal{C} = \{S_1, \dots, S_N\}$ , where each $ S_i \subseteq U $ for $i\in [N]$. The goal is to find a minimum-size set $ R \subseteq U $ such that $ R \cap S_i \neq \emptyset $ for every $ i \in [N]$.

\begin{lemma}[\cite{timothy}, Lemma 3.3]\label{lem:timothy}
Given an HS instance where $ U = \{u_1, \dots, u_n\} $ and $ \mathcal{C} = \{S_1, \dots, S_N\} $, with each subset $ S_i $ having size exactly $ \ell $, we can find a subset $ R \subseteq U $ of size $ O\left(\frac{n \log n}{\ell}\right) $ that hits all subsets in the collection, in $ O(N\ell) $ time.
\end{lemma}

The next lemma is similar to Lemma~2.3 in~\cite{ChoudharyGold}. The original lemma gives a randomized linear-time algorithm, while in the following, we state the deterministic version, which requires a higher polynomial running time. For the sake of completeness, we provide a complete proof of the lemma.
\begin{lemma}\label{lem:dom}
     Let $G=(V,E)$ be an unweighted graph on $n$ vertices. For any $\delta \in (0,1)$ and any integer $n_2\in [n]$, we can compute in polynomial-time a $\langle \lfloor \delta  D_G \rfloor , \lceil (1-\delta) D_G \rceil \rangle$-dominating set-pair of size-bound $\langle n_1,n_2\rangle$, where $n_1\leq \frac{n\log n}{n_2}$.
\end{lemma}

\begin{proof}
    Let $G=(V,E)$ be a graph, let $\delta \in (0,1) $, and let $ n_2 \in [n]$ be an integer.

    For every vertex $ v \in V $, we compute a shortest path tree $ T_v $ rooted at $v$, and let $ N_v $ denote the subset of $ V $ consisting of the $ n_2 $ vertices closest to $ v $ in $ T_v $, with ties broken arbitrarily. Specifically, $ N_v = \{ v_1, v_2, \dots, v_{n_2} \} $, where for every $v_i\in N_v, u\in V\setminus N_v$ it holds $d_{T_v} (v_i,v)\leq d_{T_v} (u,v)$.
    
    Let $ S_1 $ be the hitting-set of size $ O\left(\frac{n \log n}{n_2}\right) $, computed using the algorithm in Lemma~\ref{lem:timothy} on the collection of subsets $ \{ N_v \}_{v \in V} $.
    
    Let $ w $ be one of the farthest vertices from $ S_1 $, that is, $ d_G(S_1, w) \geq d_G(S_1, v) $ for all $ v \in V $. Fix $ S_2 = N_w $ and observe that the set-pair $ (S_1, S_2) $ has the correct size-bound $ \langle n_1, n_2 \rangle $, where $ n_1 $ satisfies the bound $ n_1 \leq \frac{n \log n}{n_2} $.

    It remains to prove that one of the two conditions of Definition~\ref{def:dsp} is satisfied. If for every $ v \in V $ we have that $ d_G(S_1, v) \leq \lfloor \delta D_G \rfloor $, then Condition~1 is satisfied and $ (S_1, S_2) $ is a dominating set-pair. Therefore, we may assume that Condition~1 is violated by at least one vertex, and in particular, $ d_G(S_1, w) > \lfloor \delta D_G \rfloor $ by the choice of $ w $.

    Let $ M_w $ denote the subset of vertices at distance at most $ \lfloor \delta D_G \rfloor $ from $ w $ in $ T_w $, i.e., $ M_w = \{ v \mid d_{T_w}(v, w) \leq \lfloor \delta D_G \rfloor \} $. Observe that $ M_w \cap S_1 = \emptyset $ because $ w $ is at distance greater than $ \lfloor \delta D_G \rfloor $ from $ S_1 $. Moreover, $ M_w \subsetneq N_w $, as otherwise $ N_w \subseteq M_w $ would imply that $ N_w \cap S_1 = \emptyset $, contradicting $ S_1 $ being a hitting-set. 

    In conclusion, since $ M_w \subset N_w = S_2 $ and the height of $ T_w $ is at most $ D_G $, we have that for all $ v \in V \setminus S_2 $, the distance is upper-bounded as $ d_G(v, S_2) \leq d_{T_w}(v, S_2) \leq D_G - \lfloor \delta D_G \rfloor = \lceil (1 - \delta) D_G \rceil $. Thus, Condition~2 of Definition~\ref{def:dsp} is satisfied.
\end{proof}

Before proving Theorem~\ref{thm:mal-approx-constr-a}, we provide an informal description of the underlying algorithm, in which we will ignore the number of used labels and other details.

In part~(i) of Theorem~\ref{thm:mal-approx-constr-a}, we set the parameters of Lemma~\ref{lem:dom} so as to ensure the existence of a set $S$ of size $O(\sqrt{n \log n})$, such that every vertex outside $S$ is reachable via a path of length at most $\lceil \frac{1}{2} \cdot D_G \rceil$. Next, for each vertex $v \in S$, we compute a shortest path tree rooted at $v$ and assign labels to its edges such that every vertex in the tree can reach the root $v$ through a temporal path whose last label is at most $D_G$. Consequently, after this step, every vertex in $G$ can temporally reach every vertex in $S$. The final step is to compute a shortest path forest (as guaranteed by Lemma~\ref{lem:dom}) rooted at $S$ that covers all vertices outside $S$ via temporal paths of length at most $\lceil \frac{1}{2} \cdot D_G \rceil$. We then assign labels to the edges of this forest using labels in the range $D_G + 1, D_G + 2, \dots, \lceil \frac{3}{2} D_G \rceil$, so that the forest can be traversed from the roots to the leaves. Finally, observe that any vertex can temporally reach all vertices in $S$, and from there, using the forest $F$, it can reach every vertex outside $S$. Thus, the resulting temporal graph is temporally connected.

In part~(ii) of Theorem~\ref{thm:mal-approx-constr-a}, we set the parameters of Lemma~\ref{lem:dom} so as to ensure the existence of either a set $S_1$ that can reach all vertices via a path of length at most $\lceil \frac{1}{3} \cdot D_G \rceil$ or a set $S_2$ that can reach all vertices via a path of length at most $\lceil \frac{2}{3} \cdot D_G \rceil$. In the second case, we proceed in the same way as in part~(i). In the first case, we compute a shortest path forest $F$ rooted at set $S_1$ and assign two labels to each edge of $F$ in such a way that every vertex in the forest can reach a vertex in $S_1$ using labels in the range $1, 2, \dots, \lceil \frac{1}{3} \cdot D_G \rceil$. Furthermore, for every non-root vertex in the forest, there exists a vertex in $S_1$ that can reach it using labels in the range $\lceil \frac{4}{3} \cdot D_G \rceil + 1, \lceil \frac{4}{3} \cdot D_G \rceil + 2, \dots, \lceil \frac{5}{3} \cdot D_G \rceil$. Finally, to ensure temporal connectivity, we compute a shortest path for every unordered pair of distinct vertices in $S_1$ and assign labels to the paths in such a way that the two vertices can reach each other using labels in the range $\lceil \frac{1}{3} \cdot D_G \rceil + 1, \lceil \frac{1}{3} \cdot D_G \rceil + 2, \dots, \lceil \frac{4}{3} \cdot D_G \rceil$. In conclusion, a vertex not in $S_1$ can temporally reach a vertex in $S_1$, then reach all vertices in $S_1$, and finally reach all vertices outside of $S_1$.

We restate Theorem~\ref{thm:mal-approx-constr-a} below and proceed with its proof.
\malapprox*
\begin{proof}
$i$)
We apply the algorithm in Lemma \ref{lem:dom} with $n_2= \lceil\sqrt{n\log n}\rceil$, and $\delta = 1/2$. With this choice of parameters we obtain a $\langle \lfloor 1/2 \! \cdot\! D_G  \rfloor , \lceil 1/2 \! \cdot\! D_G \rceil \rangle$-dominating set-pair $(S_1,S_2)$ where both $S_1$ and $S_2$ have size $O(\sqrt{n\log n})$ and one of them can reach every other vertex with a path of length bounded by $\lceil 1/2 \! \cdot\! D_G \rceil$. Let us denote with $S$ the set satisfying one of the two conditions in Definition~\ref{def:dsp}. The labeling $\lambda$ is then constructed in the following way.

For every vertex $v\in S$ we compute a shortest path tree $T_v$ rooted at $v$. Let $e=\{u,w\}$ be an edge of $T_v$ with $d_{T_v}(v,u)=d_{T_v}(v,w)+1$, we add to $\lambda(e)$ the label $D_G-d_{T_v}(v,u)+1$. Observe that, with this assignment of labels, each vertex $w\in V(T_v)$ can temporally reach the root $v$ through a temporal path ending with label $D_G$.
Since this assignment procedure is repeated for all $v\in S$, then every vertex $w\in V(G)$ can reach all the vertices in $S$ through a temporal path  ending with label $D_G$. 
Moreover, we have added $n-1$ labels for every $v\in S$, which account for an overall number of $O(n\sqrt{n\log n})$ labels.

We then compute a single shortest path forest $F$ rooted in set $S$. Let $e=\{u,w\}$ be an edge of $F$ with the vertex $u$ being the one farther from $S$, formally $\min_{v\in S} d_{F}(v,u)=\min_{v\in S}d_{F}(v,w)+1$. We add to $\lambda(e)$ the label $D_G+\min_{v\in S}d_{F}(v,u)$. Observe that the vertices at distance one from $S$ receive label $D_G+1$, the vertices at distance two receive label $D_G+2$ and so on up to the most distant edges that receive label at most $\lceil 3/2 D_G\rceil$. Moreover, we have added to $\lambda$ one label for each edge of $F$, thus an overall number of  less than $n$ labels.

In conclusion, the size of $\lambda$ is $O(n\sqrt{n\log n})$. As the minimum number of labels to temporally connect a graph, even without  constraints on the maximum allowed age, is $2n-4$ (see~\cite{KMMS}), then $|\lambda^*|\geq 2n-4$
and  $|\lambda| = O(\sqrt{n\log n})\cdot |\lambda^*|$. Moreover, the largest label used is at most  $\lceil 3/2 D_G\rceil$.

It remains to prove that $(G,\lambda)$ is temporally connected. Let $x,y\in V(G)$ be two vertices of $G$, we will prove that $x$ can reach $y$ through a temporal path. If $y\in S$ then the temporal path is equal to the unique path from $x$ to $y$ in $T_y$ using labels $D_G-d_{T_y}(y,x)+1,\dots,D_G$. Therefore, we may assume that $y\in V(G)\setminus S$. Let  $v$ be the vertex in $S$ such that  there is a path in $F$ from $v$ to $y$, this vertex must exist by construction on $F$.
If $v = x$ then the temporal path from $x$ to $y$ is the one using the path from $x$ to $y$ in $F$ and labels $D_G+1,\dots,D_G+d_{F}(v,y)$. Otherwise $v\not = x$, and the path from $x$ to $y$ is constructed by joining the temporal path from $x$ to $v$ in $T_v$ using labels $D_G-d_{T_v}(v,x)+1,\dots,D_G$ with the temporal path from $v$ to $y$ in $F$ using  labels $D_G+1,\dots,D_G+d_{F}(v,y)$, see Figure~\ref{fig:proof_tp0}.

\begin{figure}[t]
    \centering
    \pgfdeclarelayer{nodelayer}
    \pgfdeclarelayer{edgelayer}
    \pgfsetlayers{nodelayer,edgelayer}

    \includegraphics[width=0.8 \textwidth]{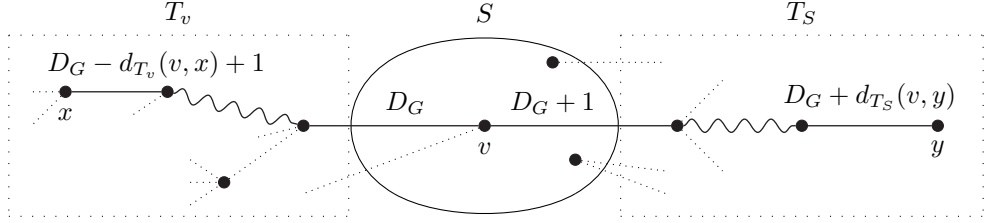}
    \caption{The temporal path from $x$ to $y$ in the proof of Theorem~\ref{thm:mal-approx-constr-a}~(i)}
    \label{fig:proof_tp0}
\end{figure}

$ii$)
Let $(S_1,S_2)$ be computed by the algorithm in Lemma \ref{lem:dom} with parameters  $n_2=\Big\lceil\sqrt[3]{D_G n \log^2 n}\Big\rceil$ and $\delta = 1/3$. We obtain a 
$\langle \lfloor 1/3D_G  \rfloor , \lceil 2/3D_G \rceil \rangle$-dominating set-pair $(S_1,S_2)$ of size-bound $\langle n_1,n_2\rangle$, where 
$n_1\leq\sqrt[3]{ D_G^{-1} n^2 \log n}$.
If Condition~2 of Definition~\ref{def:dsp} holds for $(S_1,S_2)$, then $\lambda$ is constructed as in the previous proof with $S$ being $S_2$. In this case the size of $\lambda$ is now upper-bounded by $O(n \sqrt[3]{D_G n \log^2 n})$ which is $O(\sqrt[3]{D_G n \log^2 n})\cdot |\lambda^*|$.

Therefore, we assume in the remainder that $S_1$ is such that for all $v\in V(G)$, $d_G(v,S_1)\leq \lfloor 1/3D_G \rfloor\leq 1/3D_G$.
For simplicity of presentation we will assume that $D_G$ is a multiple of $3$. 
The labeling $\lambda$ is then constructed as follows.

Let $F$ be a shortest path forest rooted in set $S_1$. Let $e=\{u,w\}$ be an edge of $F$ with the vertex $u$ being the one further from $S_1$, formally $\min_{v\in S_1} d_{F}(v,u)=\min_{v\in S_1}d_{F}(v,w)+1$. We add to $\lambda(e)$ the two labels $1/3D_G-\min_{v\in S}d_{F}(v,u)+1$ and $4/3D_G+\min_{v\in S}d_{F}(v,u)$. Observe that the vertices at distance one from $S$ receive labels $1/3D_G,4/3D_G +1$, the vertices at distance two receive labels $1/3D_G-1,4/3D_G +2$ and so on, up to the most distant edges that receive labels $1,5/3D_G$. We have added to $\lambda$ two labels for each edge of $F$, which overall do not exceed $2n$.

Let $s,s'\in S_1$ be two distinct vertices, let $P_{s,s'}=e_1,e_2,\dots,e_k$ be a shortest path between $s$ and $s'$ in $G$, and observe that $k\leq D_G$ as $G$ has diameter $D_G$. We add to $\lambda(e_i)$ the labels $1/3D_G + i$ and $1/3D_G +k- i +1$ for $i\in [k]$. Notice that with these labels we have created a temporal path from $s$ to $s'$, and a temporal path from $s'$ to $s$, both using labels in the range $1/3D_G+1,1/3D_G+2,\dots,4/3D_G$. Since we add at most $2D_G$ labels for every pair of vertices in $S_1$, then we have added a total of at most $2D_G\cdot |S_1|\cdot(|S_1|-1)/2 = O(n \sqrt[3]{D_G n \log^2 n})$ labels in this phase.
The overall number of labels, taking into account also  the previous labeling, is $|\lambda|=O(n \sqrt[3]{D_G n \log^2 n}) = O(\sqrt[3]{D_G n \log^2 n})\cdot |\lambda^*|$.  Moreover, the largest label used is at most  $5/3D_G$.

It remains to prove that $(G,\lambda)$ is temporally connected. To this aim, let $x,y\in V(G)$ be two vertices of $G$, we will prove that $x$ can reach $y$ through a temporal path.

If $x\not \in S_1$ then the vertex $x$ using the edges of $F$ can reach a vertex $v\in S_1$ through a temporal path using labels in $\{1,2,\dots,1/3D_G\}$, otherwise $x$ is already in $S_1$ and we set $v=x$. If $y\in S_1$, then adding to the previous temporal path the edges and times of $P_{v,y}$ form a temporal path from $x$ to $y$ using labels in $\{1,2,\dots,4/3D_G\}$. If $y\not \in S_1$, let $v'$ be the vertex in $S_1$ that can reach $y$ using the labels $\{4/3D_G + 1,4/3D_G + 2,\dots,5/3D_G\}$. Then, the temporal path from $x$ to $y$ is constructed using first the temporal path from $x$ to $v$ with labels in $\{1,2,\dots,1/3D_G\}$, then the temporal path $P_{v,v'}$ from $v$ to $v'$ with labels in $\{1/3D_G+1,1/3D_G+2,\dots,4/3D_G\}$, and finally the temporal path from $v'$ to $y$ that uses labels in $\{4/3D_G + 1,4/3D_G + 2,\dots,5/3D_G\}$.
\end{proof}

\section{Extensions and Variants}
In this section, we consider extensions and variants of \MAL. We begin with a discussion of the directed version of \MAL (\DMAL), and then proceed to examine two variants introduced in \cite{KMMS}, known as the Minimum Steiner Labeling (MSL) and the Minimum Aged Steiner Labeling (MASL) problems.

\subsection{\DMAL}
Consider the directed version of \MAL, denoted \DMAL, where the input graph is directed. The approximation hardness of \DMAL is related to that of \MAL in the following sense.

\begin{theorem}\label{thm:dirmalapprox}
    Given an $\alpha$-approximation algorithm for \DMAL, there exists a $2 \alpha$-approximation algorithm for \MAL.
\end{theorem}

\begin{proof}
    Let $(G,a)$ be an instance of \MAL. We construct the directed graph $G_d$ by replacing each edge $\{u,v\}$ of $G$ with the two directed edges $(u,v)$ and $(v,u)$. Observe that the number of labels in an optimal solution $\lambda^{*}_d$ for the \DMAL instance $(G_d,a)$ is at most twice that of an optimal solution $\lambda^{*}$ for the \MAL instance $(G,a)$, as we can construct a feasible solution for $(G_d,a)$ by doubling the labels of $\lambda^{*}$.

    Let $\lambda_d$ be the solution produced by the $\alpha$-approximation algorithm on the instance $(G_d,a)$ of \DMAL.


    We now construct $\lambda$ for the undirected instance as follows: for every $u,v \in V$, define $\lambda(\{u,v\}) = \lambda_d((u,v)) \cup \lambda_d((v,u))$. Clearly, $|\lambda| \leq |\lambda_d|$.

    To show that $\lambda$ is a valid solution for \MAL, observe that all directed temporal paths are preserved in the transformation. That is, if there exists a directed temporal path from $u$ to $v$ in $G'$, then there exists an undirected temporal path from $u$ to $v$ in $G$ using the same labels.

    In conclusion, we obtain that $\lambda$ is a solution for \MAL of size at most $\alpha \lambda^*_d \leq 2 \alpha \lambda^*$.
\end{proof}

By combining Theorem~\ref{thm:malinapprox}, Corollary~\ref{cor:malinapprox}  and Corollary~\ref{cor:stronginapproxmal} with Theorem~\ref{thm:dirmalapprox}, we obtain the following corollary.

\begin{corollary}\label{cor:dirMALhard}
    \begin{enumerate}[i)]
        \item For every $\epsilon\in(0,1/8)$, there is no polynomial-time $(\epsilon \log |V|)$-approximation algorithm for \DMAL, unless $\text{P} = \text{NP}$. The hardness holds even when the maximum allowed age $a$ is equal to 2.

        \item   For every $\epsilon\in(0,1/(12a))$, there is no polynomial-time $(\epsilon \log |V|)$-approximation algorithm for \DMAL, unless $\text{P}\!=\!\text{NP}$. The hardness holds even when the maximum allowed age $a$ is a fixed parameter greater or equal to 3.

        \item  For any constant $\epsilon\in (0,1)$, there is no polynomial-time $2^{\log ^{1-\epsilon} n}$-approximation algorithm for \DMAL, unless $\text{NP}\subseteq \text{DTIME}(2^{\text{polylog}(n)})$. The hardness holds even when the maximum allowed age $a$ is a fixed parameter greater or equal to 3.
    \end{enumerate}
\end{corollary}

On the approximation algorithm side, we observe that Corollary~\ref{cor:approxmal-via-dss}~(ii) also holds for \DMAL, since the result of Theorem~\ref{thm:pairwisespanner}~\cite{CDKL} applies to directed graphs as well.

Moreover, the results in Theorem~\ref{thm:mal-approx-constr-a} can also be adapted to \DMAL, achieving the same asymptotic approximation factor. To illustrate this, we provide an overview of the necessary modifications, starting with the definition of a dominating set pair in the directed setting, as introduced in Definition 2.2 of~\cite{ChoudharyGold}. From here on, we use $d_G(u,v)$ to denote the length of the shortest directed path from vertex $u$ to vertex $v$. The distances from a set to a vertex and from a vertex to a set are defined accordingly.

\begin{definition}[Definition 2.2 in~\cite{ChoudharyGold}]
    For a directed graph $G=(V,E)$ and a set-pair $(S_1,S_2)$, where $S_1,S_2\subseteq V$, we say that $(S_1,S_2)$ is a $\langle h_1,h_2\rangle$-dominating set-pair of size-bound $\langle n_1,n_2\rangle$ if $|S_1|=O(n_1)$, $|S_2|=O(n_2)$, and one of the following conditions holds.
    \begin{enumerate}
        \item For each $v\in V(G)$, $d_G(S_1,v)\leq h_1$,
        \item For each $v\in V(G)$, $d_G(v,S_2)\leq h_2$.
    \end{enumerate}
    We say that $S_1$ is $h_1$-out-dominating if it satisfies condition 1, and $S_2$ is $h_2$-in-dominating if it satisfies condition 2.
\end{definition}
From here on, we use the suffixes ``in-'' and ``out-'' to indicate that the shortest path tree (shortest path forest, respectively) is oriented toward the root(s) or toward the leaf vertices, respectively.
The proof of Lemma~\ref{lem:dom} can be adapted to establish the same result in the setting of directed dominating set-pairs. Specifically, for every $v \in V$, we compute an in-SPT $T_v$, that is, a shortest path tree in which every vertex has a directed path to the root $v$. The rest of the proof remains valid, as it is already consistent with the directed case. 

Regarding Theorem~\ref{thm:mal-approx-constr-a}~(i), we compute the directed dominating set-pair $(S_1, S_2)$ using the same values for $\delta$ and $n_2$. If $S_1$ is $\frac{1}{2}$-out-dominating, then for each $v \in S_1$ we compute an in-SPT and assign labels in the range $1, 2, \dots, D_G$ to its edges (in the following, when we assign labels, we do so in such a way that directed paths correspond to temporal directed paths). We then compute an out-SPF rooted at $S_1$ and assign to its edges labels in the range $D_G + 1, D_G + 2, \dots, \frac{3}{2}D_G$.

Otherwise, $S_2$ is $\frac{1}{2}$-in-dominating. We first compute an in-SPF rooted at $S_2$, assigning labels in the range $1, 2, \dots, \frac{1}{2}D_G$. Then, for every vertex in $S_2$, we compute an out-SPT and assign labels in the range $\frac{1}{2}D_G + 1, \frac{1}{2}D_G + 2, \dots, \frac{3}{2}D_G$ to its edges.

For Theorem~\ref{thm:mal-approx-constr-a}~(ii), we need to compute two directed dominating set pairs. The first dominating set pair $(S_1,S_2)$ is computed with $\delta=\frac{1}{3}$ and $n_2=\sqrt[3]{D_G n \log^2 n}$, while in the second dominating set pair $(S_1',S_2')$, we set $\delta=\frac{2}{3}$ and $n_2=\sqrt[3]{D_G^{-1} n^2 \log n}$. If either $S_1$ is $\frac{2}{3}$-out-dominating or $S_2'$ is $\frac{2}{3}$-in-dominating, then the labeling is constructed following the previous modification for Theorem~\ref{thm:mal-approx-constr-a}~(i). Thus, we may assume that $S_2$ is $\frac{1}{3}$-in-dominating and $S_1'$ is $\frac{1}{3}$-out-dominating. In this case, we compute an in-SPF rooted at $S_2$ and assign labels in the range $1, 2, \dots, \frac{1}{3}D_G$ to its edges. Then, for every pair of vertices $v_1 \in S_2$ and $v_2 \in S_1'$, we compute a shortest path from $v_1$ to $v_2$ and use labels in $\frac{1}{3}D_G + 1, \frac{1}{3}D_G + 2, \dots, \frac{4}{3}D_G$ to make it temporal. Finally, we compute an out-SPF rooted at $S_1'$ and, using labels in $\frac{4}{3}D_G + 1, \frac{4}{3}D_G + 2, \dots, \frac{5}{3}D_G$, we make it temporal. To summarize, let $v \in V$; either $v \in S_2$ or $v$ can temporally reach a vertex in $S_2$. Then, it can temporally reach all vertices in $S_1'$, and finally, it can temporally reach any vertex not in $S_1'$.

\subsection{MSL and MASL}
Following~\cite{KMMS}, we define the Minimum Steiner Labeling (MSL) and the Minimum Aged Steiner Labeling (MASL) problems. We begin with the definition of MSL:

\begin{definition}
    Given a graph $G=(V,E)$ and a subset $S \subseteq V$, the Minimum Steiner Labeling (MSL) problem asks for a function $\lambda : E \rightarrow 2^{\mathbb{N}}$ such that for every ordered pair $s_1, s_2 \in S$, there exists a temporal path from $s_1$ to $s_2$ in $(G,\lambda)$, and $|\lambda|$ is minimized.
\end{definition}

The authors of~\cite{KMMS} proved that \MSL is NP-complete. On the approximation side, observe that the size of an optimal Steiner Tree is at most the size of an optimal solution to \MSL, as otherwise discarding the labels would yield a smaller Steiner Tree, contradicting optimality. Moreover, the optimum of \MSL is less than twice the optimum of Steiner Tree, since, given a Steiner Tree, a feasible \MSL solution can be constructed by using the folklore labeling technique (described in the proof of Theorem~\ref{thm:malradiusdiameter}~(i)), where the root is an arbitrary vertex of the tree. As a consequence, by combining any $\alpha$-approximation algorithm for the Steiner Tree problem with the folklore labeling technique, we obtain a $2\alpha$-approximation algorithm for \MSL. To the best of our knowledge, the current best approximation ratio for the Steiner Tree problem is due to Byrka et al.~\cite{SteinerTree}, who devised a $(\ln(4) + \varepsilon) < 1.39$-approximation algorithm. Therefore, \MSL can be approximated within a factor of $2.78$.

By introducing a bound on the maximum age in the \MSL problem, we obtain the Minimum Aged Steiner Labeling (\MASL) problem.

\begin{definition}
    Given a graph $G=(V,E)$, a subset $S \subseteq V$, and an integer $a$, the Minimum Aged Steiner Labeling (\MASL) problem asks for a function $\lambda : E \rightarrow 2^{\mathbb{N}}$ such that for every ordered pair $s_1, s_2 \in S$, there exists a temporal path from $s_1$ to $s_2$ in $(G,\lambda)$, $L_\lambda \leq a$, and $|\lambda|$ is minimized.
\end{definition}

It is worth noting that when $S = V$, \MASL coincides with \MAL, and therefore the hardness of approximation results for \MAL also apply to MASL. On the approximation side, Corollary~\ref{cor:approxmal-via-dss}~(ii) holds for MASL as well, since we can use the result of Theorem~\ref{thm:pairwisespanner} (Theorem 1.8 in~\cite{CDKL}) by setting $\Delta(s,t)=a$ for $s,t\in S$, and $\Delta(s,t)=\infty$ otherwise. However, the approximation factors of the other algorithms for \MAL are affected by the change in the lower bound on the optimum. Thus, when $|S| \in \Omega(n)$, all our approximation algorithms for \MAL apply to MASL with the same asymptotic approximation factor.

\section{Conclusion and Open Questions}
In this paper, we studied the complexity of approximating the Minumum Aged Labeling Problem (\MAL). We showed that, when $a=D_G$, the problem is hard to approximate, even within large approximation factors. Table~\ref{tab:hardness} summarizes our hardness bounds. If we relax the value of $a$ and consider the cases when $a>D_G$, \MAL can be approximated within bounded factors that depend on a relation between $a$ and $D_G$. Table~\ref{tab:approx} summarizes our approximation results. We further showed a relation between \MAL and the Diameter Constrained Spanning Subgraph problem (\DCSS), which implies that similar hardness and approximation bounds apply to \DCSS.

It remains open whether \MAL is NP-hard when $D_G < a < 2D_G + 2$. We observe that our reductions are no longer valid for $a \geq D_G + 1$, as in this case it is possible to find ``small'' solutions by exploiting the presence of vertices that keep the diameter of the graph as low as the distance between the hard part of the instance used in the reduction.
Moreover, it remains unclear whether the lower bounds on the approximation of \MAL given in Theorem~\ref{thm:malinapprox} and Corollary~\ref{cor:stronginapproxmal} still hold when $D_G < a < 2R_G$. Observe that for $a \geq 2R_G$, Theorem~\ref{thm:malradiusdiameter}~(i) implies that for every constant $c > 1$, there exists an $n_0$ such that for all graphs $G$ of size at least $n_0$, we can approximate \MAL on $(G, a)$ within a factor of $c$.

\bibliography{full-paper}
\end{document}